\documentclass[article,12pt]{amsart}

\usepackage{amsthm}
\usepackage{amssymb}
\usepackage{amstext}
\usepackage{amsmath}
\usepackage{url}
\usepackage[colorlinks]{hyperref}
\hypersetup{
linkcolor=blue,
citecolor=blue,
}
\usepackage{dsfont} 
\usepackage[capitalize]{cleveref}
\usepackage[margin = 40pt,font=small,labelfont=bf]{caption}
\usepackage{arydshln}
\usepackage{pdfsync}
\usepackage[T1]{fontenc}
\usepackage{color}
\usepackage{algorithm}
\usepackage{algorithmic}
\usepackage{enumerate}
\usepackage{bbm}
\usepackage{ae,aecompl}
\usepackage{pdfpages}
\usepackage{epstopdf}
\usepackage{graphicx}
\usepackage{epsfig}
\usepackage{subfigure}
\usepackage{stmaryrd}
\usepackage{cite}

\usepackage{titletoc}

\numberwithin{equation}{section}
\theoremstyle{plain}

\newtheorem{rem}{Remark}[section]
\newtheorem{prop}{Proposition}[section]
\newtheorem{cor}{Corollary}[section]

\newtheorem{definition}{Definition}[section]

\def\build#1_#2^#3{\mathrel{\mathop{\kern 0pt#1}\limits_{#2}^{#3}}}

\newcommand{\bigO}{\ensuremath{\mathcal O}}
\def\videbox{\mathbin{\vbox{\hrule\hbox{\vrule height1.4ex \kern.6em\vrule height1.4ex}\hrule}}}
\def\demend{\hfill $\videbox$\\}

\newcommand{\RR}{{\mathbb R}}
\newcommand{\EE}{\ensuremath{\mathbb E}}

\newcommand{\XX}{\ensuremath{\Sph}}

\newcommand{\ee}{\ensuremath{\varepsilon}}
\newcommand{\SL}{\sum\limits_{l =0}^\infty \sum\limits_{m=-l}^l }
\newcommand{\tht}{\theta}

\newcommand{\Sph}{{\mathbb{S}^{2}}}
\newcommand{\Sc}{{\mathbb{S}}}
\newcommand{\Exp}{\ensuremath{\text{Exp}}}
\newcommand{\Log}{\ensuremath{\text{Log}}}
\newcommand{\QQ}{\ensuremath{\mathbf Q}}
\newcommand{\uu}{\ensuremath{\mathbf u}}
\newcommand{\FF}{\ensuremath{\mathbf F}}

\def\argmin{\mathop{\rm arg \; min}\limits}%
\def\argmax{\mathop{\rm arg \; max}\limits}%

\newcommand{\thefont}[2]{\fontsize{#1}{#2}\fontshape{n}\selectfont}
\newcommand{\1}{\rlap{\thefont{10pt}{12pt}1}\kern.16em\rlap{\thefont{11pt}{13.2pt}1}\kern.4em}

\begin{document}

\title[]{ Regularized estimation of Monge-Kantorovich quantiles for spherical data \vspace{1ex}}
\author{Bernard Bercu, J\'{e}r\'{e}mie Bigot and Gauthier Thurin}
\dedicatory{\normalsize Universit\'e de Bordeaux \\
Institut de Math\'ematiques de Bordeaux et CNRS  (UMR 5251)}
\thanks{The authors gratefully acknowledge financial support from the Agence Nationale de la Recherche  (MaSDOL grant ANR-19-CE23-0017).}

\maketitle

\thispagestyle{empty}

\begin{abstract}
Tools from optimal transport (OT) theory have recently been used to define a notion of quantile function for directional data.
In practice, regularization is mandatory for applications that require out-of-sample estimates. 
To this end, we introduce a regularized estimator built from entropic optimal transport, by extending the definition of the entropic map to the spherical setting.
We propose a stochastic algorithm to directly solve a continuous OT problem between the uniform distribution and a target distribution, by expanding Kantorovich potentials in the basis of spherical harmonics. 
In addition, we define the directional Monge-Kantorovich depth, a companion concept for OT-based quantiles. 
We show that it benefits from desirable properties related to Liu-Zuo-Serfling axioms for the statistical analysis of directional data. 
Building on our regularized estimators, we illustrate the benefits of our methodology for data analysis.

\end{abstract}

\



\section{Introduction}

\subsection{Quantiles for directional data using optimal transport}

In various situations, data are naturally associated to directions that belong to the circle or the unit $d$-sphere $\Sc^{d-1}$ for $d \geq 2$. Such observations, referred to as directional data, can be found in  various applications including wildfires \cite{ameijeiras2018directional}, gene expressions \cite{dortet2008model}, or cosmology \cite{marinucci2008spherical} to name but a few. 
Directional statistics \cite{mardia2000directional,ley2017modern,ley2018applied,pewsey2021recent} is the field that brings together the corresponding models, methods and applications for statistical inference.

In this paper, we focus on the concept of quantiles for directional data. For real random variables, the notion of quantile is a well established statistical concept thanks to the canonical ordering of observations on the real line. 
Beyond the setting of distributions with rotational symmetry \cite{ley2014new}, the absence of a canonical ordering on $\Sc^{d-1}$ makes the definition of quantiles for directional data  more involved.

A recent line of research in nonparametric statistics \cite{chernozhukov2015mongekantorovich,Hallin-AOS_2021,hallin2022nonparametric,HallinReport2021}  deals with the use of the theory of optimal transport (OT) to define Monge-Kantorovich (MK) quantiles\footnote{These are also referred to as center-outward quantiles, \cite{Hallin-AOS_2021}. } for multivariate data. Desirable properties of MK quantiles include ancillarity, distribution-freeness of associated ranks, and consistence with the univariate setting \cite{Hallin-AOS_2021}, together with connections to the celebrated Tukey's notion of statistical depth \cite{chernozhukov2015mongekantorovich} and Mahalanobis distance \cite{hallin2018mahalanobis}. 
The concept of MK quantiles has also proven to be fruitful in many applications, including statistical testing \cite{ghosal2021multivariate,HallinTests2022,Ziang2022,shi2021,shi2022distribution}, regression \cite{Carlier:2022wq,delBarrio2022}, risk measurement \cite{beirlant2019centeroutward}, and Lorenz maps \cite{Fan2022,Hallin_mordant_2022}.

This approach has recently been  applied in  \cite{hallin2022nonparametric} to obtain a new notion of  center-outward quantiles  for directional data. 
Starting from independent and identically distributed ($i.i.d.$) directional data $X_1,\cdots,X_n$ sampled from a target measure $\nu$, the main idea in \cite{hallin2022nonparametric} is to define an empirical quantile function $\QQ_n$ as an optimal matching.
This is obtained by solving an OT problem between the empirical measure $\widehat{\nu}_n = \frac{1}{n} \sum_{i=1}^n \delta_{X_i}$ and a uniform distribution $\mu_n = \frac{1}{n} \sum_{i=1}^n \delta_{U_i}$ based on a regular grid $\{U_1,\cdots,U_n\}$ over the unit sphere $\Sc^{d-1}$. 
In this manner, if $U$ is sampled from ${\mu}_n$, 
the random vector $\QQ_n(U)$ follows the empirical distribution of the observations, which is consistent with the standard univariate quantile function.
On the other hand, the ranks of $X_1,\cdots,X_n$ follow the uniform distribution ${\mu}_n$, independently from $\nu$, which means that they are \textit{distribution-free}.
Building upon this property, \cite{hallin2022nonparametric} investigate rank-based statistical testing.
However, the resulting 
matching between two discrete distributions does not provide out-of-sample estimates. 

In contrast, this paper focuses on applications of MK quantiles where one is interested in observations $X_{\rm new}$ different from $X_1,\cdots,X_n$. 
We give two examples that motivate our upcoming contributions. 
For depth-based classification, the depth (or outlyingness) of a testing point $X_{\rm new}$ must be computed after learning on training data $X_1,\cdots,X_n$. 
For dispersion measurements, volumes of quantile regions $\mathbb{C}_\tau$ can be estimated by rejection sampling, repeatedly considering the event $\{X_{\rm new}\in \mathbb{C}_\tau\}$.

\subsection{Main contributions}

To compute out-of-sample estimates, we pursue regularized estimation of directional MK quantiles. 
To do so, we rely on entropically regularized OT \cite{cuturi2013sinkhorn}, that is well-known for its computational advantages. 
While the computational cost of solving a discrete OT problem is known to scale cubically in the number of observations \cite{peyre2020computational}, entropic regularization improves this to quadratic complexity.

\textbf{A stochastic algorithm.}
We suggest to adapt the stochastic algorithm developed in \cite{BBT2023} dealing with multivariate data in $\RR^d$, that is not assumed to belong to the unit $d$-sphere. 
By expanding Kantorovich dual potentials in their series of Fourier coefficients, each iteration of the stochastic algorithm in   \cite{BBT2023} reduces to the use of the Fast Fourier Transform (FFT). 
On the sphere $\Sph$, dual potentials can be parameterized via spherical harmonics coefficients instead, that is the analog of the Fourier basis for square-integrable functions on $\Sph$. 
In this manner, using a sequence of random variables $X_1,\ldots,X_n$ sampled from a target distribution $\nu$ supported on $\Sph$, we construct a stochastic algorithm in the space of spherical harmonics coefficients.
In practice, this algorithm depends on the choice of a grid of points in $\Sph$ of size $\bigO\left( p^2 \right)$ to implement a FFT on $\Sph$  \cite{wieczorek2018shtools}.
The computational cost at each iteration is thus of order $\bigO\left( p^2 \log^2(p)\right)$  \cite{KUNIS200375}. 
A key motivation for such stochastic algorithm is its ability to solve OT in an online fashion, where data arrives sequentially. 
Also, it overcomes the need to discretize the reference distribution, by learning directly the continuous OT problem.

\textbf{A regularized quantile function.}
Furthermore, we also derive an estimator of the quantile function for spherical data that is smoother than $\QQ_n$.
This regularized quantile function is a directional counterpart of the \textit{entropic map} \cite{pooladian2021entropic}, following classical results on entropic optimal transport in the euclidean space $\RR^d$. 
 Its main benefit, beyond computational complexity, is to allow for out-of-sample estimation. 

\textbf{Directional statistical depth.}
Finally, we introduce the new notion of MK directional statistical depth, in accordance with the euclidean MK depth \cite{chernozhukov2015mongekantorovich}. 
We discuss its properties relative to traditional Liu-Zuo-Serfling axioms for the statistical analysis of directional data.
Moreover, we study statistical applications built from it, and provide a comparison with other estimators, to better highlight the potential of entropic regularization for spherical quantiles estimation and out-of-sample estimation.

\subsection{Organization of the paper}

In Section \ref{sec:related}, we discuss related works on alternative definitions and estimators of quantiles for directional data. The definitions of OT-based directional distribution and quantile functions on $\Sph$  are recalled in Section \ref{sec:MainDefinitions}. Our approach to obtain regularized estimators of MK quantiles on $\Sph$ is detailed in Section \ref{sec:regest}. 
A study of the MK directional statistical depth, is proposed in Section \ref{depthstat} and illustrated with simulated data. 
Numerical experiments are reported in Section \ref{numexp} to highlight the benefits of entropic regularization. 
Concluding remarks and discussions on this work are proposed in Section \ref{sec:diss}. Finally, mathematical details and proofs are deferred to technical appendices.
Python codes for the experiments carried out
in this paper are available at
\url{https://github.com/gauthierthurin/SphericalQuantiles}.

\section{Related works} \label{sec:related}

\subsection{Directional quantiles and depth}

Very much related to multivariate quantiles is the idea of statistical depth and the center-outward ordering of data, with a long history dating back to Tukey's work \cite{tukey75}.
The issue of \textit{picturing data}, in Tukey's terminology, has since gained very much attention, which prevents us  from providing an exhaustive survey, and we
rather refer to \cite{liu1999multivariate,barnett1976ordering,serfling2006depth,mosler2013depth}.  
The first notion of directional depth function was introduced in \cite{small1987measures}, followed by the work of \cite{liu1992ordering} that developed three different approaches. 
The properties of the latter have been studied and applied for inference in \cite{rousseeuw2004characterizing,agostinelli2013nonparametric}.
The required computational effort led the authors of \cite{ley2014new} to build the angular Mahalanobis depth, and, doing so, they provided the first concept of directional quantiles.
Despite appealing properties, the obtained contours are constrained to be rotationally symmetric, motivating the elliptic counterpart from \cite{hauch2022quantiles}.
Still, the elliptic assumption is a strong one as discussed in \cite{hallin2022nonparametric}.
Facing either this lack of adaptiveness or the computational burdens of previous references, distance-based depths were proposed in \cite{pandolfo2018distance}, even though not explicitly related to the notion of quantiles. 
These directional depth functions can been applied, for instance, in data analysis and inference \cite{liu1992ordering,small1987measures,agostinelli2013nonparametric,konen2023spatial}, classification \cite{pandolfo2018distance,liu1992ordering,demni2019cosine,Demni2021,Nagy2023,konen2023spatial} or clustering \cite{Pandolfo2023}. 
Among recent years, two concepts of multivariate quantiles have emerged in $\RR^d$, namely the spatial quantiles \cite{chaudhuri1996geometric} and the center-outward ones \cite{chernozhukov2015mongekantorovich}, both gathering most of commonly sought-after properties.
More importantly here, these promising ideas have successfully been extended to directional data \cite{konen2023spatial,hallin2022nonparametric}, improving on the lack of adaptiveness of Mahalanobis quantiles \cite{ley2014new}, with desirable asymptotic results inherited from the formalism of quantiles. 

 To put it in a nutshell, existing concepts of directional quantiles include Mahalanobis quantiles, Spatial quantiles and Monge-Kantorovich ones.
On the one hand, a statistical depth associated to a notion of quantiles is amenable to benefit from the best of both worlds, that is adaptivity to the underlying geometry and consistency of empirical versions, as argued for instance in \cite{konen2023spatial}.  
On  another hand, in comparison with other directional quantiles, Monge-Kantorovich ones present an additional descriptive power inherited from the fact that $\QQ(U)\sim\nu$. 
A direct consequence is that $\QQ$ must contain all the available information, which is appealing with the purpose of summing up unknown features of multivariate data. 
Lastly, we mention that the estimation of density level sets on $\mathbb{S}^{d-1}$, \cite{saavedra2022nonparametric}, is a related but different objective.
Indeed, the density allows to capture local information such as multimodality (as with a histogram) while quantiles focus on a global centrality information (as with a boxplot).

\subsection{Regularized OT maps on the sphere}

In \cite{hallin2022nonparametric}, the estimation of directional quantiles amounts to a discrete OT problem. 
However, such estimator is piecewise constant, restricted to take its values in the set of  observations. 
There, regularization naturally enters the picture, as highlighted in $\RR^d$ in several works, either after the estimation of unregularized OT \cite{Hallin-AOS_2021,beirlant2019centeroutward}, or with entropic optimal transport (EOT) \cite{Carlier:2022wq,BBT2023,Masud2021}.
Facing the same issue in the present directional context, the entropic map is extended here to the non-Euclidean setting, similarly to what is done in \cite{cuturi2023monge} for general costs in $\RR^d$. 
To the best of our knowledge, this appears to be new, although it benefits from explicit formulation tractable in linear time, given the dual potentials solving EOT. 
This belongs to the line of work estimating OT maps on manifolds.  
A meaningful approach is to use deep-learning architectures \cite{cohen2021riemannian}.  
In particular, \cite{pegoraro2023vector} extends \cite{hallin2022nonparametric} to vector quantile regression on manifolds, with convincing numerical performance. 
Some works leverage mesh adaptation \cite{weller2016mesh} or PDE \cite{hamfeldt2022convergence}.
The recent theoretical work \cite{frungillo2024discrete} provides complete convergence result, building on an interesting idea of discrete OT maps through geometric medians, akin to the barycentric projection in $\RR^d$. 
Following insights from previous works on Euclidean data \cite{Carlier:2022wq,BBT2023,Masud2021}, 
 EOT in the quantiles' context is a complementary perspective that enforces smoothness, to fulfill our objective of depth-based data analysis.

\section{Directional distribution and quantile functions on the 2-sphere based on optimal transport}\label{sec:MainDefinitions} 

In this section, we introduce the main definitions of OT-based distribution and quantile functions for spherical data, beginning with   notation related to spherical harmonics and differentiation on $\Sph$. 

\subsection{ Context} 

The unit $2$-sphere is defined by 
$
\Sph = \{ x \in \RR^3 : \Vert x \Vert = 1 \}.
$ 
The points $x \in \Sph$ can be written in spherical coordinates, with longitude $\phi \in [ -\pi, \pi]$ and colatitude $\theta \in [0,\pi]$, as
\begin{equation}\label{Phispheric_coord}
x = \Phi(\theta,\phi) := (\cos \phi \sin \theta, \sin \phi\sin\theta,\cos\theta).
\end{equation}
On the $2$-sphere, the geodesic distance is 
$
d(x,y)= \arccos (\langle x,y\rangle),
$
and the squared Riemannian distance is
\begin{equation}\label{def_cost}
c(x,y) = \frac{1}{2} d(x,y)^2.
\end{equation}
Both $d$ and $c$ are continuous and bounded as $d(x,y)\in [0,\pi]$. Moreover,
$(\Sph,d)$ is a separable complete metric space, with Borel algebra $\mathcal{B}^2$.
The surface measure $\sigma_\Sph$ on $\Sph$ is given by 
$$
\int_\Sph f(x) d\sigma_\Sph(x) = \int_0^\pi \int_{-\pi}^\pi f(\Phi(\theta,\phi)) \sin\theta d\phi d\theta,
$$
and the uniform measure on $\Sph$ writes $\mu_\Sph = \frac{1}{4\pi}\sigma_\Sph$. 
The space of all equivalence classes of square integrable functions on $\Sph$ is denoted by 
$L^2(\Sph)$. 
We define the \textit{spherical harmonic} function of degree $l \in \mathbb{N}^*$ and order $m\in\{-l,\cdots,l\}$ by 
$$
Y_l^m(x) = Y_l^m(\Phi(\theta,\phi)) = \sqrt{\frac{2l+1}{4\pi} \frac{(l-m)!}{(l+m)!} } P_l^m(\cos \theta)e^{im\phi},
$$
where the associated Legendre functions $P_l^m: [-1,1]\rightarrow \RR$ verify, for $l\in \mathbb{N}^*$ and $m \geq 0$, 
$$
P_l^m(t) = \frac{(-1)^m}{2^l l!} (1-t^2)^{m/2} \frac{d^{l+m}(t^2-1)^l}{dt^{l+m}}
\quad
\text{and}
\quad
P_l^{-m}(t) = (-1)^m \frac{(l-m)!}{(l+m)!} P_l^m(t).
$$
Importantly, the spherical harmonics form an orthonormal basis of $L^2(\Sph)$, so that every function $f \in L^2(\Sph)$ is uniquely decomposed, for $x \in \Sph$, as
\begin{equation}\label{spheric_harmonics}
f(x) = \sum_{l=0}^\infty \sum_{m = -l}^l \bar{f}_l^m Y_l^m(x),
\end{equation}
where the sequence of spherical harmonic coefficients $\bar{f} = (\bar{f}_l^m)$ verifies
\begin{equation}\label{spheric_harmonics_coefs}
\bar{f}_l^m = \frac{1}{4\pi}\int_\Sph f(x) \overline{Y_l^m(x)} d\sigma_\Sph(x).
\end{equation}
We refer to \cite{chirikjian2000engineering} for an introduction to Fourier analysis on the sphere. 
Below, we introduce a few notation from differential geometry that one can find for instance in \cite{ferreira2014concepts}. 
At any point $x\in \Sph$, the tangent space is $\mathcal{T}_x\Sph = \{ y \in \RR^3 : \langle x,y \rangle = 0\}$, and the associated orthogonal projection $\rho_x : \RR^3 \rightarrow \mathcal{T}_x\Sph$ verifies
\begin{equation}\label{def_rhox}
\rho_x \xi = (I - xx^T)\xi = \xi - \langle \xi,x \rangle x.
\end{equation}
The exponential map at $x\in \Sph$,
$\Exp_x :\mathcal{T}_x\Sph \rightarrow \Sph$,
has the explicit form
$$
\Exp_x (v) = \cos(\Vert v \Vert) x + \sin(\Vert v \Vert)\frac{v}{\Vert v \Vert},
$$
and its inverse $\Log_x$ writes
\begin{equation}\label{Logx}
    \Log_x(z) = \frac{d(x,z)}{\sqrt{1-\langle x,z \rangle^2}}\rho_x z = d(x,z) \frac{\rho_x (z-x)}{\Vert \rho_x (z-x) \Vert}.
\end{equation}
The Riemannian gradient is given by orthogonally projecting Euclidean derivatives onto the tangent space.
For a smooth function $f : \Sph \rightarrow \RR$, its Riemannian gradient $\nabla f(x)$ at $x$ is thus defined by 
\begin{equation}\label{defRiemannianGrad}
\nabla f(x) = \rho_x D f(x), 
\hspace{1cm} \mbox{where} \hspace{1cm}
D f(x) = \Big( \frac{\partial f(x)}{\partial x_i} \Big)_{i,j \in \{1,2,3\}}.
\end{equation}
The Riemannian Hessian $\nabla^2 f(x) : \mathcal{T}_x\Sph \rightarrow \mathcal{T}_x\Sph$ at $x$ is defined by the same token,
\begin{equation}\label{defRiemannianHess}
\nabla^2 f(x) = \rho_x \Big[ D^2 f(x) - \langle D f(x),x\rangle I \Big]
\hspace{0.5cm} \mbox{with} \hspace{0.5cm}
D^2f(x) =\Big( \frac{\partial^2 f(x)}{\partial x_i \partial x_j} \Big)_{i,j \in \{1,2,3\}}.
\end{equation}

\subsection{Main definitions}\label{maindef}

Here, we fix the definitions introduced in \cite{hallin2022nonparametric}.
 Given $\mu$ and $\nu$ two probability measures supported on $\Sph$, we say that $T$ pushes forward $\mu$ to $\nu$, denoted by $T_\#\mu = \nu$, if, for each measurable $B \in \mathcal{B}^2$, 
$
\nu(B) = \mu(T^{-1}(B)). 
$
Then, for the quadratic cost $c$ defined in \eqref{def_cost}, Monge's formulation of the OT problem between $\mu$ and $\nu$ writes
\begin{equation}\label{Monge}
\min\limits_{T : T_\#\mu = \nu} \EE_{X \sim \mu} \left[c(X,T(X)) \right].
\end{equation}
The solution of \eqref{Monge} is referred to as a \textit{Monge map}, while $\mu$ and $\nu$ are called the \textit{reference} and \textit{target} measures, respectively. Before considering the existence of Monge maps,
we recall the key definition of $c$-transforms as stated in 
\cite{mccann2001polar}, that is equivalent to the formulation of  \cite{hallin2022nonparametric}[Definition 2].

\begin{definition}\label{c_transform}
Given a function $\psi : \Sph \rightarrow \RR$, its $c$-transform is defined by 
$
\psi^c(y) = \inf_{x \in \Sph} \{ c(x,y) - \psi(x) \}.
$
Then, $\psi$ is said to be $c$-concave when $\psi^{cc} := (\psi^c)^c = \psi$.
\end{definition}

The proper summary of \cite{hallin2022nonparametric}[Proposition 1] highlights that $c$-concavity is related to optimality in Monge's OT problem \eqref{Monge}.
For our continuous and bounded cost $c$, a sufficient assumption is that the reference measure belongs to $\mathbf{B}_2$, the family of $\sigma_\Sph$-absolutely continuous distributions with densities bounded away from $0$ and $\infty$, see \cite{mccann2001polar}[Theorem 9], which is the case for the uniform $\mu_\Sph$.
This enables the definition of distribution and quantile functions in \cite{hallin2022nonparametric}. 


\begin{definition}\label{def:quantile}
The directional MK quantile function of the arbitrary probability measure $\nu$ is the $\mu_\Sph$-$a.s.$ unique map $\QQ : \Sph \rightarrow \Sph$ such that $\QQ_\# \mu_\Sph = \nu$ and there exists a $c$-concave differentiable mapping $\psi : \Sph \rightarrow \RR$ such that, $\sigma_\Sph$-$a.e.$,
$$ 
\QQ(x) = \Exp_x(-\nabla \psi(x)).
$$ 
In addition, the directional MK distribution function of $\nu$ is given by
$$ \FF(x) = \Exp_x(-\nabla \psi^c(x)).$$ 
\end{definition}

As soon as $\nu$ belongs to $\mathbf{B}_2$,
\cite{mccann2001polar}[Corollary 10] ensures that $\FF=\QQ^{-1}$ almost everywhere, whereas $\QQ^{-1}$ might not exist if $\nu$ is not absolutely continuous.

\begin{rem}[Regularity]\label{rem_regularity}
The regularity of OT maps on the sphere is a delicate subject that has inspired a number of works, \cite{von2013regularity,loeper2011regularity,Delanoe2006,LoeperVillani}.
Any $c$-concave potential $\psi$ is twice differentiable almost everywhere \cite{cordero2001riemannian}[Proposition 3.14]. 
For further regularity, an appropriate requirement is that the underlying measures are smooth and belong to $\mathbf{B}_2$.
In particular, if, a minima, $\nu$ has density $f\in C^{1,1}(\Sph)$ with respect to $\sigma_\Sph$, then the MK quantile function $\QQ$ belongs to $C^{2,\beta}(\Sph)$ for all $\beta \in \,]0,1[$, \cite{LoeperVillani}. 
\end{rem}


In view of the proof of \cite{mccann2001polar}[Theorem 9], $\psi$ and its $c$-transform $\psi^c$ in Definition \ref{def:quantile} maximize the dual version of Kantorovich's problem, in the sense that
\begin{equation}\label{Kanto}
(\psi,\psi^c) \in \argmax_{(u,v)\in \text{Lip}_c} \int_\Sph u (x) d\mu_\Sph(x) +  \int_\Sph v (y) d\nu(y),
\end{equation}
with 
$
\text{Lip}_c = \left\{ u,v : \Sph \rightarrow \RR \text{ continuous } ; u(x) + v(y) \leq c(x,y) \right\}. 
$
The semi-dual version of \eqref{Kanto} refers to the optimisation over a single dual variable $u$, the other being taken as $u^c$.
 

A central point in $\Sph$ must be chosen for the reference
 $\mu_\Sph$,
 to define
  nested regions with $\mu_\Sph$-probability $\tau\in [0,1]$. 
A well-suited notion of central point 
is the Fr\'echet median\footnote{This choice is motivated by its simplicity, and it implicitly assumes that the Fr\'echet median is unique. However, this may not be the case, even at the population level. For instance, this happens if $\nu$ charges antipodal regions equally.}
\begin{equation}\label{FrechetMed}
\theta_{M} = \argmin_{z \in \Sph} \EE_{Z\sim\nu}[d(Z,z)],
\end{equation}
that can be computed with the package \textit{geomstats} \cite{JMLR:v21:19-027}. 
Then, the spherical cap with $\mu_\Sph$-probability $\tau\in [0,1]$ centered at $\FF(\theta_M)$ is 
$$
\mathbb{C}^U_\tau = \left\{ x \in \Sph : \langle x, \FF(\theta_M)\rangle \geq 1 - 2\tau \right\},
$$
with boundary $\mathcal{C}^U_\tau = \left\{ x \in \Sph : \langle x, \FF(\theta_M)\rangle = 1 - 2\tau \right\}$ a \textit{parallel} of order $\tau$. 
This defines a rotated version of the usual latitude-longitude coordinate system \eqref{Phispheric_coord}, with respect to the pole $\FF(\theta_M)$, as follows. 
Any $x \in \Sph$ decomposes into
\begin{equation}\label{tangent_normal}
x = \langle x, \FF(\theta_M)\rangle \FF(\theta_M) + \sqrt{1 - \langle x, \FF(\theta_M)\rangle^2} \mathbf{S}_{\FF(\theta_M)}(x),
\end{equation}
where $\langle x, \FF(\theta_M)\rangle$ is a latitude\footnote{ The example $\FF(\theta_M) = (0,0,1)^T$ eases the understanding: $x\in \mathcal{C}_\tau^U$ implies $\langle x, \FF(\theta_M) \rangle = x_3 = 1 - 2\tau$ and the latitude is fixed.} on the parallel $\mathcal{C}^U_\tau$, while the \textit{directional sign}
\begin{equation}\label{directional_signs}
 \mathbf{S}_{\FF(\theta_M)}(x) = \frac{x - \langle x, \FF(\theta_M)\rangle \FF(\theta_M)}{ \Vert x - \langle x, \FF(\theta_M)\rangle \FF(\theta_M)  \Vert}
\end{equation}
is a longitude, with convention $\mathbf{0}/0 = \mathbf{0}$ for $x = \pm \FF(\theta_M)$. 
$ \mathbf{S}_{\FF(\theta_M)}(x)$ takes values on the rotated equator $\mathcal{C}^U_{1/2}$, and thus characterizes
meridians crossing $s \in\mathcal{C}^U_{1/2}$ through $\mathcal{M}_s^U = \{ x\in \Sph : \mathbf{S}_{\FF(\theta_M)}(x) = s \}$.
The image by $\QQ$ of the parallel / meridian system \eqref{tangent_normal} provides curvilinear parallels $\QQ(\mathcal{C}^U_\tau)$ and curvilinear meridians $\QQ(\mathcal{M}_s^U)$ adapted to the geometry of the support of $\nu$, giving rise to suitable directional concepts of quantile contours and signs \cite{hallin2022nonparametric}.
Intuitively, a change in coordinates in a data-adaptive fashion must retain all the available information, in a simpler form amenable to be summed up. 
 
\begin{definition}[Quantile contours, regions and signs]
\label{Qcontours}
Let $\nu \in \mathbf{B}_2$, with directional quantile function $\QQ$. Then, 
\begin{itemize}
\item the quantile contour of order $\tau \in [0,1]$ is $\mathcal{C}_\tau = \QQ(\mathcal{C}^U_\tau)$,
\item the quantile region of order $\tau \in [0,1]$ is $\mathbb{C}_\tau = \QQ(\mathbb{C}^U_\tau)$,
\item the sign curve associated with $s\in \mathcal{C}^U_{1/2}$ is $\QQ(\mathcal{M}_s^U)$.
\end{itemize}
\end{definition}
Since $\QQ$ is a push-forward mapping, $\QQ_\# \mu_\Sph = \nu$, the $\nu$-probability content of $\mathbb{C}_\tau$ is $\tau$. Moreover, the quantile contours of $\nu \in \mathbf{B}_2$ are continuous and the quantile regions are closed, connected and nested, as stated in \cite{hallin2022nonparametric}.

\section{Regularized estimation}
\label{sec:regest}

\subsection{Entropic OT on the 2-sphere}\label{sec:algorithm} 

We now introduce an algorithm based on the spherical Fourier transform to solve the regularized Kantorovich problem on the $2$-sphere. 
In the Euclidean case, Kantorovich's problem is known to be easier to solve than Monge's problem \cite{peyre2020computational}. 
Even more so, adding an entropic regularization term to \eqref{Kanto} has been a cornerstone for the development of OT-based methods in statistics, \cite{cuturi2013sinkhorn}. 
For online learning, \cite{aude2016stochastic} proposed stochastic algorithms by rewriting the dual objective function, with convergence guarantees \cite{bercu2020asymptotic}.
For arbitrary measures, (not only discrete ones),  the dual variables cannot be viewed as finite-dimensional vectors anymore. 
Therefore, one requires the use of nonparametric families of dual functions.
For Euclidean data in $\RR^d$, one can use reproducing kernel Hilbert spaces \cite{aude2016stochastic}, deep neural networks \cite{Seguy2018} or Fourier series \cite{BBT2023}. 
For directional data in $\mathbb{S}^{d-1}$, we leverage spherical harmonics.
The resulting algorithm directly targets the continuous OT problem between the continuous reference measure and the underlying $\nu$ in an online fashion, instead of the (semi) discrete OT problem towards the empirical measure $\widehat{\nu}_n$.
Also, the Network simplex and Sinkhorn algorithm require the storage of the cost matrix, of size $n^2$ for two samples of sizes $n$, whereas stochastic algorithms are designed to avoid it. 
 
For spherical quantiles, one must solve an OT problem where the reference measure is $\mu_\Sph$, the uniform probability measure on the sphere. 
The semi-dual version of EOT between $\mu_\Sph$ and $\nu$ from \cite{genevay:tel-02458044}[Proposition 12], for $\ee>0$ a regularization parameter, writes
\begin{equation}\label{def:eot}
\max_{u \in L^{\infty}(\Sph)} \int_{\Sph} u(x) d\mu_\Sph(x) +  \int_{\Sph} u^{c,\varepsilon}(y) d\nu(y),
\end{equation}
with 
$u^{c,\varepsilon}$ the \textit{smooth conjugate} of $u$ defined by
\begin{equation}
u^{c,\ee}(y) = -\ee \log \left( \int_{\Sph} \exp \Big( \frac{u(x) -c(x,y)}{\varepsilon} \Big) d\mu_\Sph(x) \right).
 \label{eq:reg-c-transform}
\end{equation}
The smooth conjugate is the entropic counterpart of Definition \ref{c_transform}. 
For bounded costs, the problem \eqref{def:eot} admits a solution in $L^{\infty}$, unique up to additive constants,  \cite{genevay:tel-02458044}[Theorem 7].
To leverage unicity, we impose that $\int_{\Sph} u(x) d\mu_\Sph(x) = 0$, so that the optimisation problem \eqref{def:eot} becomes
$$
\max_{u \in L^{\infty}(\Sph)}     \int_{\Sph} u^{c,\varepsilon}(y) d\nu(y).
$$
It is well-known,\cite{Nutz2022entropic,aude2016stochastic}, that $\uu_\ee$ is solution of \eqref{def:eot} if and only if 
\begin{equation}\label{charac_c_transf}
    \uu_\ee = ((\uu_\ee)^{c,\ee})^{c,\ee}.
\end{equation}

\textbf{Rewriting the problem.}
A Lipschitz continuous function on $\Sph$ equals its spherical Fourier series \eqref{spheric_harmonics} pointwise \cite{michel2012lectures}[Theorem 5.26]. 
One can find in \cite{mccann2001polar}[Lemma 2] that the true unregularized potentials are Lipschitz, because $\Sph$ has a finite diameter $\vert \Sph \vert = \pi$.
The same holds in the regularized case $\ee>0$, from the optimality condition \eqref{charac_c_transf}, see Proposition 12 from \cite{feydy2019interpolating}[Appendix B] or \cite{Nutz2022entropic}[Lemma 3.1].
Consequently, we suggest to parameterize the dual variable in \eqref{def:eot} by its spherical harmonic coefficients.
For a given $\varepsilon > 0$, we consider the optimal sequence of coefficients $\bar{\uu}_\ee$ defined as the solution of the following {\it stochastic convex minimisation} problem
\begin{equation}
\bar{\uu}_\ee = \argmin_{\bar{\uu} \in \ell_{1} }  H_\ee(\bar{\uu}) \hspace{1cm} \mbox{with}  \hspace{1cm}  H_\ee(\bar{\uu}) = \mathbb{E} \left[ h_\ee(\bar{\uu},X) \right] \label{Se}
\end{equation}
where $X$ is a random vector with distribution $\nu$ , $\bar{\uu} = (\bar{\uu}_l^m)_{l \geq 1} $  and
$$
h_\ee(\bar{\uu},x) = - u^{c,\ee}(x) \quad \mbox{with} \quad u(z) = \SL \bar{\uu}_l^m Y_l^m(z).
$$
Note that the spherical harmonic coefficient $\bar{\uu}_0^0$ equals $0$ because of the identifiability condition $\int_{\Sph} u(x) d\mu_\Sph(x) = 0$.

We shall now discuss the equivalence between \eqref{Se} and the original problem \eqref{def:eot}. 
On the $2$-sphere, the series of spherical harmonics of a continuously differentiable function is uniformly convergent, see \cite{kellogg2012foundations}[p.259] or \cite{kalf1995expansion}. 
To obtain the stronger result that the sequence of spherical harmonics belongs to $\ell_1$, the function needs to be twice continuously differentiable, \cite{kalf1995expansion}[Theorem 2].
For the unregularized potentials, such differentiability requires smoothness of the measures involved, as highlighted in Remark \ref{rem_regularity},
but this  
holds for the regularized potential without any continuity condition for $\nu$. 

\begin{prop}\label{uu_ee_continu_deriv}
Let $\uu_\ee$ be a solution of \eqref{def:eot}. 
Then, $\uu_\ee$ is twice continuously differentiable, and, as a byproduct, its series of spherical harmonics belongs to $\ell_1$. 
\end{prop}

The proof of the above result is detailed in our supplementary material.
Consequently, the problem \eqref{Se} is equivalent to the original one \eqref{def:eot}. The main virtue of this parameterization is that
partial derivatives of $h_\ee$ with respect to the parameters $\bar{\uu}_l^m$ can be derived easily, which is appealing in view of a stochastic gradient scheme. 
Here, the objective $h_\ee(\cdot, x)$ is of the same mathematical nature than in \cite{BBT2023}, so that it is differentiable and the following property holds.

\begin{prop}\label{prop:grad}
For every $x \in \Sph$, the function $h_\ee(\cdot,x) : \ell_{1} \rightarrow \mathbb{R}$ is Fr\'echet differentiable and its differential $D_{\bar{\uu}} h_\ee(\bar{\uu}, x)$ belongs to the dual Banach space $(\ell_{\infty}, \| \cdot \|_{\ell_{\infty}})$ where
$
 \| \bar{\uu} \|_{\ell_{\infty}} = \sup_{l,m} \vert \bar{\uu}_l^m\vert.
$
The components of $D_{\bar{\uu}} h_\ee(\bar{\uu}, x)$ are the partial derivatives   
\begin{equation}
\label{partialderh}
\frac{\partial h_\ee(\bar{\uu},x) }{\partial \bar{\uu}_l^m} =  \frac{1}{4 \pi} \int_\Sph g_{\bar{\uu},x}(z) Y_l^m(z)  d\sigma_\Sph(z),
\end{equation}
that are the spherical harmonics coefficients of the function 
\begin{equation}\label{def_F_tht}
g_{\bar{\uu},x}(z) = \frac{\exp \left( \frac{u(z) -c(z,x)}{\ee} \right) }{ \int \exp \left( \frac{u(y)-c(y,x)}{\ee} \right) d\mu_\Sph(y)} 
\hspace{0.5cm} \mbox{with} \hspace{0.5cm} u(z) = \SL \bar{\uu}_l^m Y_l^m(z).
\end{equation}
\end{prop}

\textbf{The algorithm.}
For $(X_n)$ a sequence of independent random vectors with distribution $\nu$, we consider the stochastic algorithm in the Banach space $(\ell_{1}, \| \cdot \|_{\ell_1})$ defined, for
$n \geq 0$, by
\begin{equation}
\widehat{u}_{n+1} = \widehat{u}_{n} - \gamma_n W D_{\widehat{u}} h_\ee(\widehat{u}_{n} , X_{n+1})
\label{defthetan}
\end{equation}
where $\gamma_n = \gamma n^{-\alpha}$ is a decreasing sequence of positive numbers with $1/2 < \alpha  < 1$ and $\gamma > 0$.
Because $\ell_{1}$ differs from its dual space $\ell_{\infty}$, the linear operator $W$ is defined by
$$
\left\{\begin{array}{ccc}
W :  (\ell_{\infty}, \| \cdot \|_{\ell_{\infty}})  & \to &  (\ell_{1}, \| \cdot \|_{\ell_1}) \\
\bar{v} = (\bar{v}_l^m)  & \mapsto & \bar{w} \odot \bar{v} = (\bar{w}_l^m \bar{v}_l^m)
\end{array}\right.
$$
where $\bar{w} = (\bar{w}_l^m)$ is a {\it deterministic sequence of positive weights} satisfying the condition
\begin{equation}
\label{condw}
\|\bar{w}\|_{\ell_1} = \SL \bar{w}_l^m < + \infty.
\end{equation}
In all the experiments carried out in this paper, we use $\bar{w} _l^m = ( l^2+m^2 )^{-1}$.
A {\it regularized estimator} of the optimal potential $\uu_\ee(x) = \SL \bar{\uu}_{\ee,l}^m Y_l^m(x)$ is naturally given by 
\begin{equation}
\label{DEFun}
\widehat{\uu}_{\varepsilon,n}(x) = \SL \widehat{u}_{n,l}^m Y_l^m(x).
\end{equation}

From a practical point of view, the stochastic sequence \eqref{defthetan} must be discretized. 
To do so, one must consider a grid of $p^2$ points on the $2$-sphere and the associated spherical harmonics coefficients.  
We emphasize that this discretization takes place in the space of frequencies, willing to take advantage from implicit interpolation in this space.
The Python library \textit{pyshtools} \cite{wieczorek2018shtools}, implements spherical harmonics transforms and reconstructions. 
Our numerical procedure builds upon it as the stochastic algorithm \eqref{defthetan} requires computing the spherical harmonics coefficients of the function $g_{\widehat{u},y}(x)$ in \eqref{def_F_tht}, that relies on $u(x)$ reconstructed from $\widehat{u}$ thanks to the inverse Fourier transform on the $2$-sphere. 
With the help of the fast routine within \textit{pyshtools}, the computational cost at each iteration of \eqref{defthetan} is of order $\bigO\left( p^2 \log^2(p)\right)$.

 \begin{rem}
 To study the convergence of the  stochastic algorithm \eqref{defthetan}, we could adapt the theoretical results in \cite{BBT2023} to our setting.
  Importantly, the convergence results (almost-sure convergence of estimators towards the minimizer in $\ell^2$-norm) obtained therein remain valid in the present spherical context, because they were irrespective of the orthonormal basis and the cost $c$.  
  However, this is beyond the scope of this paper, that aims to showcase benefits of EOT for out-of-sample estimation.
 \end{rem}

\subsection{Regularized distribution and quantile functions} \label{sec:quantilefunction} 

We now introduce the regularized counterpart of Definition \ref{def:quantile}.
When dealing with measures supported on $\RR^d$, one can use the \textit{entropic map}, defined as the barycentric projection of the entropic optimal plan, \cite{pooladian2021entropic}. 
At first sight, this requires a notion of average on the sphere, as done in \cite{frungillo2024discrete} from the unregularized empirical OT plan. 
Nonetheless, this map is alternatively characterized by analogy with Brenier theorem \cite{pooladian2021entropic}, whose building block is the gradient of Kantorovich potential, as in \cite{cuturi2023monge} for OT with a general cost in $\RR^d$.
Because this enforces the structure of optimality, we pursue this idea for our non-Euclidean setting. 

\begin{definition}\label{def:entropicFQ}
Let $\nu$ be an arbitrary probability measure supported on $\Sph$, and $\uu_\ee : \Sph \rightarrow \RR$ be a solution of \eqref{def:eot} between $\mu_\Sph$ and $\nu$. 
Then, the regularized distribution function of $\nu$ is given by 
\begin{equation}\label{eq:defFee}
\FF_\ee (z) = \Exp_z(-\nabla \uu_\ee^{c,\ee}(z)),
\end{equation}
and the regularized quantile function of $\nu$ is 
\begin{equation}\label{eq:defQee}
\QQ_\ee (x) = \Exp_x(-\nabla \uu_\ee(x)).
\end{equation}
\end{definition}

\noindent 
This requires the differentiation of entropic Kantorovich potentials. For a given regularization parameter $\varepsilon > 0$, partial derivatives can be retrieved by
$$
\frac{\partial \uu_{\varepsilon}(x)}{\partial x_i}  = \SL \bar{\uu}_{l}^m \frac{\partial  Y_l^m(x)}{\partial x_i},
$$
where the package  \textit{sphericart} \cite{bigi2023fast}, allows the computation of $\frac{\partial  Y_l^m(x)}{\partial x_i}$ easily. The Riemannian gradient $\nabla \uu_\ee$ follows using \eqref{defRiemannianGrad}.
But because this may lead to numerical instabilities, we suggest instead to make use of first-order conditions \eqref{charac_c_transf}. 
Explicit forms for the generalized entropic maps on the hypersphere are given below.
Their proof in the supplementary material follows by differentiating  \eqref{charac_c_transf}.

\begin{prop}\label{entropicQeeFee}
The regularized distribution and quantile functions of $\nu$ on $\Sph$ admit closed-form expressions through
$$
\QQ_\ee(x) = \Exp_x \int \Log_x(z) \exp \Big( \frac{\uu_\ee(x) - c(x,z) + \uu_\ee^{c,\ee}(z)}{\ee} \Big) d\nu(z),
$$
and
$$
\FF_\ee(z) = \Exp_x \int \Log_z(x) \exp \Big( \frac{\uu_\ee(x) - c(x,z) + \uu_\ee^{c,\ee}(z)}{\ee} \Big) d\mu_\Sph(x).
$$
\end{prop}

\begin{rem} We stress that 
$\FF_\ee$, $resp.$ $\QQ_\ee$, does not push $\nu$ forward to $\mu_\Sph$ anymore, $resp.$ $\mu_\Sph$ forward to $\nu$. 
However, they are expected to be close to their unregularized counterparts, for small values of $\ee>0$, as studied, for the quadratic cost in $\RR^d$, in \cite{goldfeld2022limit,pooladian2021entropic,stromme2023sampling}. 
The limit $\ee \rightarrow 0$ has been studied outside the Euclidean setting \cite{EOTgeometry2022,Nutz2022entropic}, although not directly about the generalized entropic map itself. 
In particular, up to some sequence $(\ee_k)$ such that $\lim_{k\rightarrow +\infty}\ee_k = 0$, \cite{Nutz2022entropic}[Proposition 3.2] gives us the uniform convergence of potentials $(\uu_{\ee_k},\uu_{\ee_k}^{c,{\ee_k}})$ on compact subsets of $\Sph$, towards $(\psi,\psi^c)$ solving \eqref{Kanto}.
\end{rem}

From Proposition \ref{entropicQeeFee}, $\FF_\ee$ and $\QQ_\ee$ can be seen as weighted averages in the tangent space.
This can be gainful in practice, because of the regularity it induces.
Indeed, second-order derivatives are given in our supplementary material, which entails continuity for $\FF_\ee$ and $\QQ_\ee$. 
Finally, we stress that
\begin{equation*}
(x,z) \mapsto \exp \Big( \frac{\uu_\ee(x) - c(x,z) + \uu_\ee^{c,\ee}(z)}{\ee} \Big)
\end{equation*}
is the density of the optimal entropic plan with respect to $\mu \otimes \nu$, \cite{Nutz2022entropic}. Besides, 
by properties of $\exp$ and \eqref{charac_c_transf}, it equals,  
\begin{align}\label{equiv_form_g_eps}
\frac{\exp \Big( \frac{\uu_\ee(x) - c(x,z)}{\ee} \Big)}{
\int \exp \Big( \frac{\uu_\ee(y) - c(z,y)}{\ee} \Big) d\mu_\Sph(y)}=  \frac{\exp \Big( \frac{\uu_\ee^{c,\ee}(z) - c(x,z)}{\ee} \Big)}{
\int \exp \Big( \frac{\uu_\ee^{c,\ee}(y) - c(x,y)}{\ee} \Big) d\nu(y)}.
\end{align}

\subsection{ Regularized empirical distribution and quantile functions}\label{sec:emp}

Suppose that the estimator $\widehat{\uu}_{\ee,n}$, defined in \eqref{DEFun}, has been computed using the stochastic algorithm \eqref{defthetan} from  $i.i.d.$ observations $X_1,\cdots,X_n$ sampled from $\nu$ supported on $\Sph$.  
To obtain a regularized quantile function, the empirical counterpart of Proposition \ref{entropicQeeFee} would involve integrals with respect to $\mu_\Sph$ to compute the smooth conjugate of $\widehat{\uu}_{\ee,n}$. 
Therefore, to circumvent this issue, we consider a random sample $U_1,\cdots,U_N$ uniformly drawn on $\Sph$, and we define the following estimator (as an approximation of $\widehat{\uu}_{\ee,n}^{c,\ee}$)
\begin{equation}\label{def_discrete_c_transf}
    \widehat{\uu}_{N,n}^{c,\ee}(z) = -\ee \log \frac{1}{N}\sum_{i=1}^N \exp \Big( \frac{\widehat{\uu}_{\ee,n}(U_i) - c(U_i,z)}{\ee} \Big).
\end{equation}
Leveraging \eqref{equiv_form_g_eps}, we propose the following estimator, for the regularized quantile function $\QQ_\ee$ defined in Proposition \ref{entropicQeeFee},
\begin{equation}\label{defQepsn}
\hat{\QQ}_{N,n}^{\ee} (x) = \Exp_x\Big(\sum_{i=1}^n  \hat{g}_{N,n}^{\ee}(x,X_i)\Log_x(X_i) \Big), 
\end{equation}
where
$$
\hat{g}_{N,n}^{\ee}(x,z) = 
 \frac{\exp \Big( \frac{\widehat{\uu}_{N,n}^{c,\ee}(z) - c(x,z)}{\ee} \Big)}{ \sum_{j=1}^n \exp \Big( \frac{\widehat{\uu}_{N,n}^{c,\ee}(X_j) - c(x,X_j)}{\ee} \Big)}.
$$
In the same token, an estimator of $\FF_\ee$ is given by
\begin{equation}\label{defFepsn}
\hat{\FF}_{N,n}^{\ee}(z) = \Exp_z \Big( \sum_{i=1}^N  \tilde{g}_{N,n}^{\ee}(U_i,z) \Log_z(U_i)\Big),
\end{equation}
with
$$
\tilde{g}_{N,n}^{\ee}(x,z) = 
 \frac{\exp \Big( \frac{\widehat{\uu}_{\ee,n}(x) - c(x,z)}{\ee} \Big)}{ \sum_{j=1}^N \exp \Big( \frac{\widehat{\uu}_{\ee,n}(U_j) - c(U_j,z)}{\ee} \Big)}.
$$

\section{Depth-based data analysis}\label{depthstat}

This section is dedicated to study a companion concept of directional MK quantiles, the MK statistical depth. 
We state a directional definition and discuss its properties.
After that, we introduce descriptive tools in the spirit of  the ones presented in \cite{liu1999multivariate,beirlant2019centeroutward} in $\RR^d$. 
For the sake of completeness, we first study the Euclidean setting $\RR^d$ before the directional one, that is of particular interest for us. 
Indeed, the results that we derive below do not appear as such in the literature, at least to the best of our knowledge.

\subsection{Euclidean setting}

We begin with the main definitions taken from \cite{chernozhukov2015mongekantorovich}. 
Our chosen reference measure, denoted by $U_d$, is given by the random vector $R\Phi$, for $R$ and $\Phi$ independently drawn from $[0,1]$ and from the unit hypersphere $\mathbb{S}^{d-1} =  \{\varphi \in \RR^d : \Vert \varphi \Vert = 1 \}$, respectively, as originally proposed in \cite{chernozhukov2015mongekantorovich,Hallin-AOS_2021} to define MK quantiles in $\RR^d$. 
Note that the MK distribution function, the inverse of the MK quantile function, might not exist, $e.g.$ if $\nu$ is discrete. 
Following \cite{ghosal2021multivariate}, this is tackled with the Legendre-Fenchel dual of a convex function $\psi$, given by $\psi^*(x) = \sup_{u\in \RR^d}\{ \langle x,u\rangle -\psi(u)\}$.

\begin{definition}[\cite{chernozhukov2015mongekantorovich}]\label{defMKdepth_quantiles}
Let $\nu$ be a probability measure on $\RR^d$. 
Its MK quantile function is the unique $\QQ =\nabla \psi$ for some convex $\psi :\RR^d\rightarrow \RR$ such that $\QQ_\# U_d = \nu$. Then, the MK $\alpha$-quantile contour is 
$
\QQ (\{ u \in \RR^d : \Vert u \Vert = \alpha \}).
$
The sign curve associated to $u\in\mathbb{B}(0,1)$ is 
$
\QQ (\{ t\frac{u}{\Vert u \Vert} : t \in [0,1]\}).
$
The MK depth of $x \in \RR^d$ is the depth of $\nabla \psi^*$ under Tukey's depth, \cite{tukey75},
\begin{equation*}
D_\nu(x) = D_{U_d}^{Tukey}\Big(\nabla \psi^*(x)\Big).
\end{equation*}

\end{definition}

The Liu-Zuo-Serfling axioms \cite{liu1990notion,zuo2000general}, describe desirable properties for depth concepts. 
The MK-depth softens some of them, to reach more relevant contours \cite{chernozhukov2015mongekantorovich}. 
Firstly, MK depth corresponds to Tukey depth for elliptical families \cite{chernozhukov2015mongekantorovich}.
Moreover, it benefits from invariance properties \cite{ghosal2021multivariate}[Lemmas A.7,A.8], with respect to scaling (multiplication by a positive constant), translations, and orthogonal transformations (multiplication by an orthogonal matrix). 
Note that the affine-invariance does not hold.
Another axiom is the \textit{linear monotonicity relative to the deepest points}, that is $D_\nu(x) \leq D_\nu((1-t)x_0 + tx)$ for all $t\in[0,1]$ if $x_0$ is a deepest point. 
This is not fulfilled by the MK depth \cite{chernozhukov2015mongekantorovich}, although it verifies a similar property along sign curves, as we shall see now. 

\begin{prop}[Curvilinear monotonicity relative to the deepest points]\label{CurvilinearMonotSC}
Assume that $\nu$ is continuous. 
The MK depth is monotonically decreasing along sign curves, that is, for each $u\in \mathbb{B}(0,1)$ and $t\in[0,1]$, 
$$
D_\nu(\QQ(u)) \leq D_\nu(\QQ(tu)).
$$
\end{prop}

\begin{proof}
Recall that Tukey's depth verifies linear monotonicity relative to the deepest points \cite{zuo2000general}. 
As the origin is the deepest point for $U_d$, this writes, for any $t\in[0,1]$,
\begin{equation}\label{tukeymonoUd}
D_{U_d}^{Tukey}(u) \leq D_{U_d}^{Tukey}(tu).
\end{equation}
From Definition \ref{defMKdepth_quantiles}, $D_\nu(\QQ(u)) = D_{U_d}^{Tukey}(\nabla \psi^*\circ \QQ(u))$. 
By continuity of $\nu$, $\nabla\psi^* \circ \QQ (u) = u$ $a.e.$, \cite{Villani}[Theorem 2.12 and Corollary 2.3].
Thus the result follows, with \eqref{tukeymonoUd}.
\end{proof}

This corresponds to the classical linear monotonicity under distributions with straight sign curves, including spherical families due to the particular form of the MK quantile function in this setting, taken from \cite{chernozhukov2015mongekantorovich}. 

\begin{cor}\label{CurvilinearMonotSC_sphericdata}
For spherically symmetric distributions, sign curves are straight lines, and the MK depth verifies linear monotonicity relative to the deepest point. 
For any $x$ in the support of $\nu$, for $m$ the MK deepest point of $\nu$,
\begin{equation}\label{resSCstraightlines}
\forall t \in [0,1], D_\nu(x) \leq D_\nu((1-t)m + t x).
\end{equation}
\end{cor}

\begin{proof}
Let $X$ be a random vector associated with a spherically symmetric distribution.  
By inverting the MK distribution function known from \cite{chernozhukov2015mongekantorovich}, 
$$
\QQ(u) = \frac{u}{\Vert u \Vert}G^{-1}(\Vert u \Vert) + m,
$$
where $G$ is the univariate distribution function of the radial part $\Vert X -  m \Vert$.
Because $\Vert X \Vert \geq 0$ $a.s.$ and $G^{-1}$ is increasing, $G^{-1}(t\Vert u \Vert) / G^{-1}(\Vert u \Vert) \in [0,1]$ and
$$
\QQ(tu) = \frac{u}{\Vert u \Vert}G^{-1}(t\Vert u \Vert) + m
= \frac{G^{-1}(t\Vert u \Vert)}{G^{-1}(\Vert u \Vert)} \Big( \QQ(u)-m)\Big) + m.
$$
This rewrites, for $\delta_t =  G^{-1}(t\Vert u \Vert) / G^{-1}(\Vert u \Vert)$,
$
\QQ(tu) = \delta_t \QQ(u) + (1-\delta_t) m.
$
Besides, $\delta_t$ takes all values between $0$ and $1$ for $t\in [0,1]$. 
This, combined with Proposition \ref{CurvilinearMonotSC} induces 
\begin{equation*}
\forall u \in \mathbb{B}(0,1), \forall \delta \in [0,1], D_\nu(\QQ(u)) \leq D_\nu( \delta \QQ(u) + (1-\delta) m).
\end{equation*}
But any $x$ in the support of $\nu$ writes $\QQ(u)$ for $u = \FF(x)$, which gives \eqref{resSCstraightlines}. 
\end{proof}

\subsection{Directional setting}

Using the same ideas, one can define the MK depth on the sphere through any statistical depth with respect to the uniform $\mu_\Sph$ oriented towards $\FF(\theta_M)$, and the simplest is surely to consider the proximity with $\FF(\theta_M)$.

\begin{definition}\label{defMKdepth}
Let $\nu$ be an arbitrary probability measure on $\Sph$, with directional distribution function $\FF$.
The directional MK depth of $x\in \Sph$ is defined by 
$
    D_\nu(x) = 1 - d(\FF(x),\FF(\theta_M))/\pi.
$
\end{definition}

Regarding the $\Sph$-adapted versions of Liu-Zuo-Serfling axioms, \cite{liu1990notion,zuo2000general}, the directional MK depth behaves like its Euclidean counterpart. 
We begin with the four classical properties which are analogs of the Euclidean axioms,  \cite{ley2014new}.
The affine-invariance is replaced on $\Sph$ by rotational invariance, which holds true from the supplementary material of \cite{hallin2022nonparametric}.  
Moreover, it is straightforward that
$D_\nu$ attains its maximum at the center $\theta_M$, and that it vanishes at $x$ such that $\mathbf{F}(x) =-\mathbf{F}(\theta_M)$. 
Hence, outliers outside the support of the distribution will be ranked around $-\mathbf{F}(\theta_M)$ and will have vanishing depth.
Finally, monotonicity along great circles is not fulfilled, but it is replaced in the same data-adaptive fashion than in $\RR^d$.

\begin{prop}[Curvilinear monotonicity relative to the deepest points]\label{CurvilinearMonotSC2}
Assume that $\nu$ is continuous.
The directional MK depth is monotonically decreasing along sign curves.
For each $x\in \Sph$ and $t\in[\langle x,\FF(\theta_m)\rangle,1]$, let $x_t \in \mathcal{M}_{s}^U$, for $s=\mathbf{S}_{\FF(\theta_M)}(x)$, such that
\begin{equation}\label{eqxtSC}
x_t = t \FF(\theta_M) + \sqrt{1-t^2}s.
\end{equation}
Then,
$$
D_\nu(\QQ(x)) \leq D_\nu(\QQ(x_t)).
$$
\end{prop}

\begin{proof}
Fix $x\in \Sph$. From the decomposition \eqref{tangent_normal}, $s = \mathbf{S}_{\FF(\theta_M)}(x)$ is the directional sign associated to $x$.
For $t\in[-1,1]$, let $x_t \in \mathcal{M}_{s}^U$ be a parameterization of the reference sign curve associated to $s$, as in \eqref{eqxtSC}.
Immediately, one may note that $\langle x_t, \FF(\theta_M)\rangle = t$, and $x_t = x$ for $t=\langle x,\FF(\theta_M)\rangle$. 
Besides, $D_\nu(\QQ(x_t)) = 1 - d(x_t,\FF(\theta_M))/\pi = 1 - \arccos(t)/\pi$, so, as soon as $t \geq \langle x,\FF(\theta_m)\rangle$, 
$
D_\nu(\QQ(x_t)) \geq D_\nu(\QQ(x)).
$
\end{proof}

Explicit formulations for rotationally invariant distributions are given in \cite{hallin2022nonparametric}, and recalled in the supplementary material,
so that the following is straightforward.

\begin{cor}\label{CurvilinearMonotSC_sphericdata2}
For rotationally invariant distributions, sign curves are great circles. 
Thus, the MK depth verifies linear monotonicity along great circles, relative to the deepest point. 
\end{cor}

\begin{proof}
Without loss of generality, assume that $\tht_M =(0,0,1)^T$, which can always be ensured by applying a rotation to the data. 
Then, from the supplementary of \cite{hallin2022nonparametric}, $x\in \Sph$ and $\QQ(x)$ only differ from their latitude, under the assumption of $\nu$ to be rotationally invariant.
Thus, sign curves are great circles, and the result follows.
\end{proof}

Other desirable axioms have been put forward recently in \cite{nagy2024theoretical,Nagy2023}, namely the \textit{upper semi-continuity} and the \textit{non-rigidity of central regions}. 
\textit{Upper semi-continuity} is ensured as soon as the MK distribution function $\FF$ is continuous, thus at least for $\nu\in\mathbf{B}_2$. 
When $\nu$ is arbitrary, taking the regularized directional MK depth built from $\FF_\ee$, for $\ee>0$, imposes continuity. 
The \textit{non-rigidity of central regions} states that quantile regions are not restricted to be spherical caps, which is readily true for the MK depth.
In fact, its adaptivity to the underlying support can be seen as a stronger \textit{non-rigidity} axiom, requiring that $\QQ(U) \sim \nu$ as soon as $U\sim\mu$.
Furthermore, Proposition \ref{CurvilinearMonotSC2} and Corollary \ref{CurvilinearMonotSC_sphericdata2} shed some light on the non-verified axiom of \textit{monotony along great circles}.
Our results suggest that the directional MK depth alleviates these axioms when necessary, for complex distributions such as mixtures, whereas the axioms are fulfilled for distributions for which it is useful, in particular for rotationally invariant ones.

\section{Numerical experiments}\label{numexp}

\subsection{Choice of $\ee$ and comparison of EOT and OT}\label{mse}

The calibration of $\epsilon$ is a common issue in EOT. 
It can be selected by cross-validation, 
as with any hyper-parameter, or it can be selected with a transport-based
criterion \cite{vacher2022parameter}. 
In this section, we shall see that $\epsilon\approx0.1$ leads to accurate results overall, leading to this choice as a baseline in all the experiments of this paper.

\subsubsection{Quantitative effect of regularization}

In Figure \ref{fig_compar_trueMSE}, we study the influence of the regularization parameter on a quantitative metric: we compare against ground-truth contours of the isotropic von-Mises Fisher distribution. Indeed, this is the only distribution for which explicit formula are known. This was first established by \cite{hallin2022nonparametric} and it is recalled in the supplementary material for the sake of completeness.
The mean-squared error from uniform samples $(x_i) \subset \Sph$ is
$$
\mathcal{R}_n(\widehat{Q}) = \frac{1}{n}\sum_{i=1}^n c( Q(x_i), \widehat{Q}(x_i)),
$$
for $\widehat{Q}$ denoting either our regularized MK quantile estimator or the unregularized one proposed in \cite{hallin2022nonparametric}. 
 Estimators $\widehat{Q}$ are evaluated on
samples of size $n=500$ drawn from the uniform $\mu_\Sph$ and from a von-Mises fisher distribution of location $(0,0,1)^T$ and concentration $\kappa = 10$. 
By doing this experiment $100$ times, we obtain a boxplot of MSE values for various values of $\ee$. The results are reported in Figure \ref{fig_compar_trueMSE}, where $\ee=0$ refers to the MSE of $\widehat{\QQ}_0$. 
The dashed horizontal line illustrates the median value for the MSE of $\widehat{\QQ}_0$. 
It can be observed that entropic regularization is able to significantly outperform the estimation of the quantile map, in particular for values around $\ee \approx 0.1$.  
Increasing the sample size leads to similar conclusions. 

\begin{figure}[htbp]
\centering 

{\subfigure[Sample size $n=500$]{\includegraphics[width=0.4\textwidth]{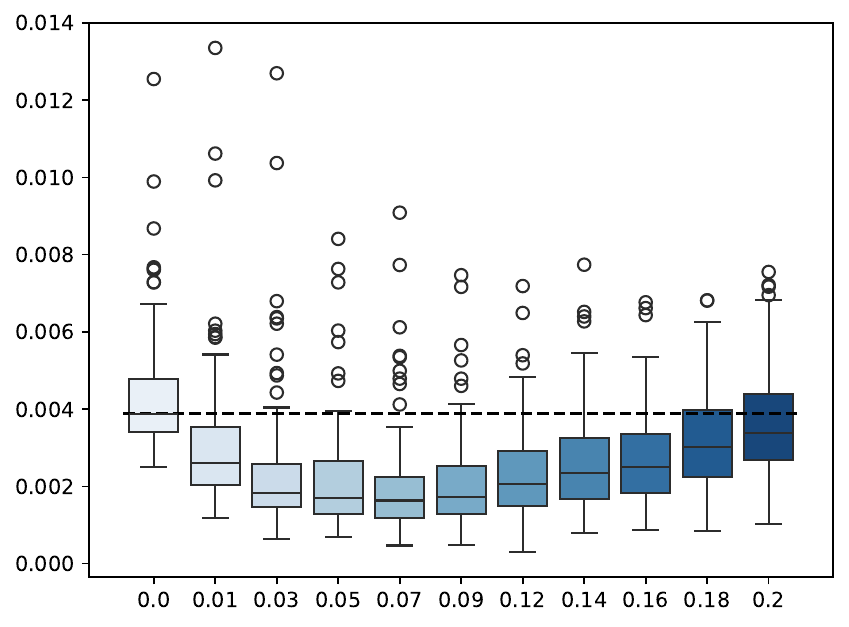}}}
{\subfigure[Sample size $n=1000$]{\includegraphics[width=0.4\textwidth]{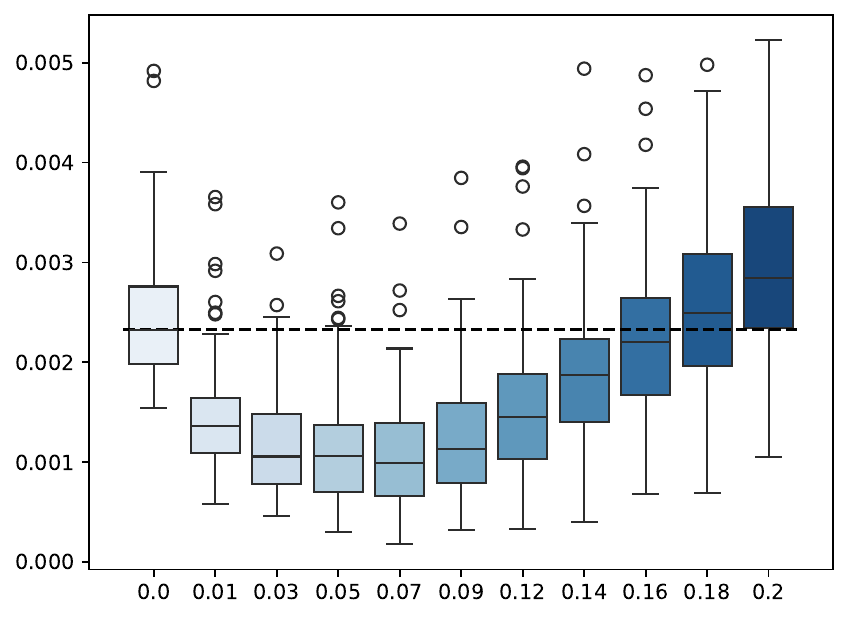}}}

\caption{Mean squared error $\mathcal{R}_n(\widehat{\QQ}_\ee)$ as a function of the regularization parameter $\ee \in [0,0.2]$. The horizontal dashed-line is the value of $\mathcal{R}_n(\widehat{\QQ}_0)$.}
\label{fig_compar_trueMSE}
\end{figure}

\subsubsection{Qualitative effect of regularization} 
In Figure \ref{fig_various_eps}, we visually compare regularized and unregularized MK spherical quantile contours of orders $24.4\%, 48.8\%, 75.6\%$ on the same von-Mises fisher distribution than in Figure \ref{fig_compar_trueMSE}.
Such uncommon probability contents are inherent to unregularized contours because the number of contours (and their size) depends on a fixed grid, and a fortiori on the sample size, that is here fixed at $n=2001$. 
Ground-truth contours are presented together with unregularized and regularized ones, for $\ee\in\{0.01,0.05,1\}$. 
Each contour contains $100$ points, linked by straight lines. 
For $\ee=0.01$, contours adapt too much on the finite-sample data, causing errors as ground-truth contours are smoother.
For $\ee=1$, contours are smoother, but there is too much bias in the approximation between $\widehat{\QQ}_\ee$ and the underlying ground truth. 
For the well-chosen $\ee=0.05$, the trade-off between regularity and low-bias allows the better estimation. 
This sheds some light on the behavior of regularization.
The lower the $\ee$, the more adapted $\widehat{\QQ}_\ee$ is to the finite-sample data and its irregularities. 
Larger values of $\ee$ induce smoother contours, as a byproduct of a greater regularity for $\widehat{\QQ}_\ee$. 

\begin{figure}[htbp]
\centering
{\subfigure[$\ee = 0.01$]{\includegraphics[width=0.3\textwidth]{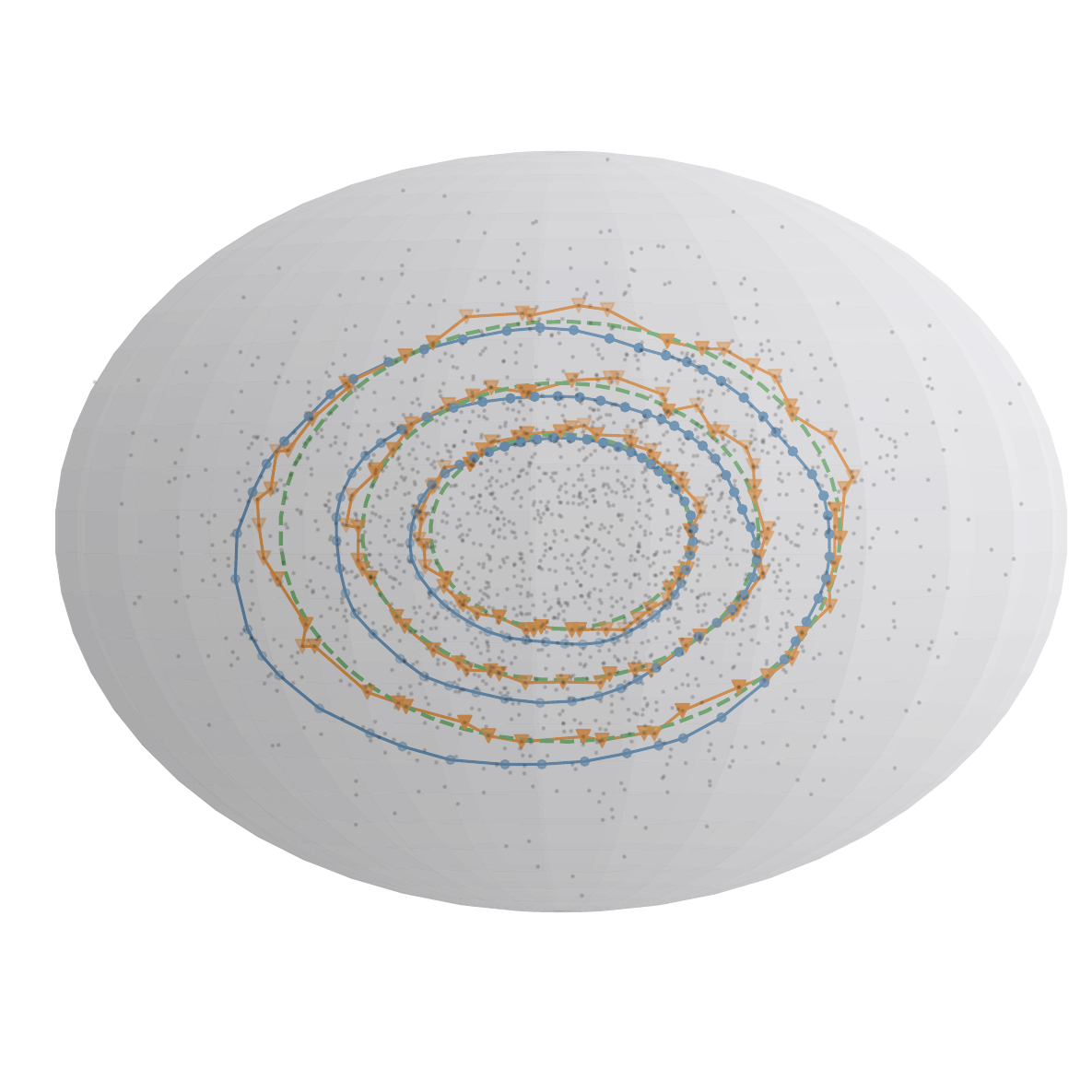}}}
{\subfigure[$\ee = 0.05$]{\includegraphics[width=0.3\textwidth]{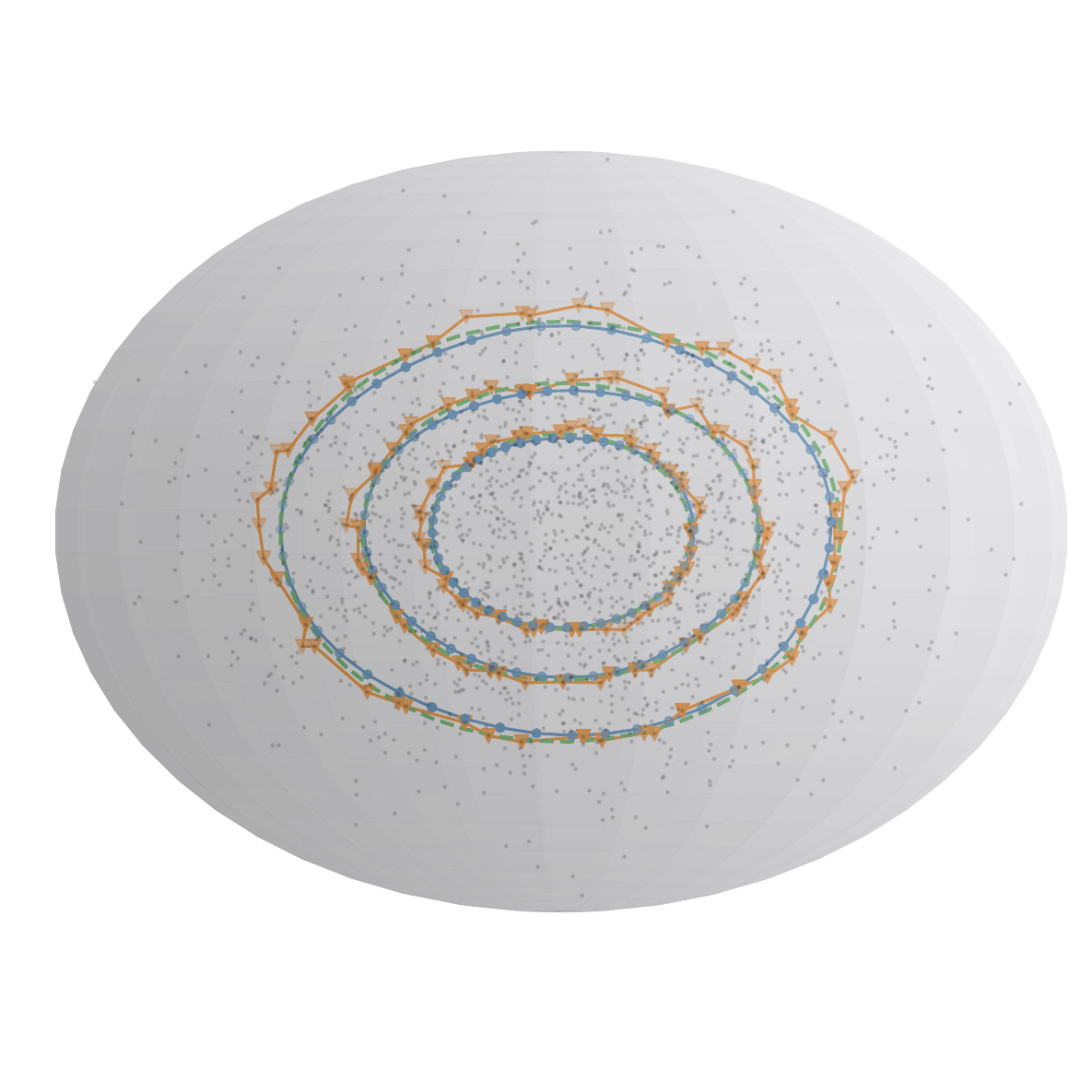}}}
{\subfigure[$\ee = 1$]{\includegraphics[width=0.3\textwidth]{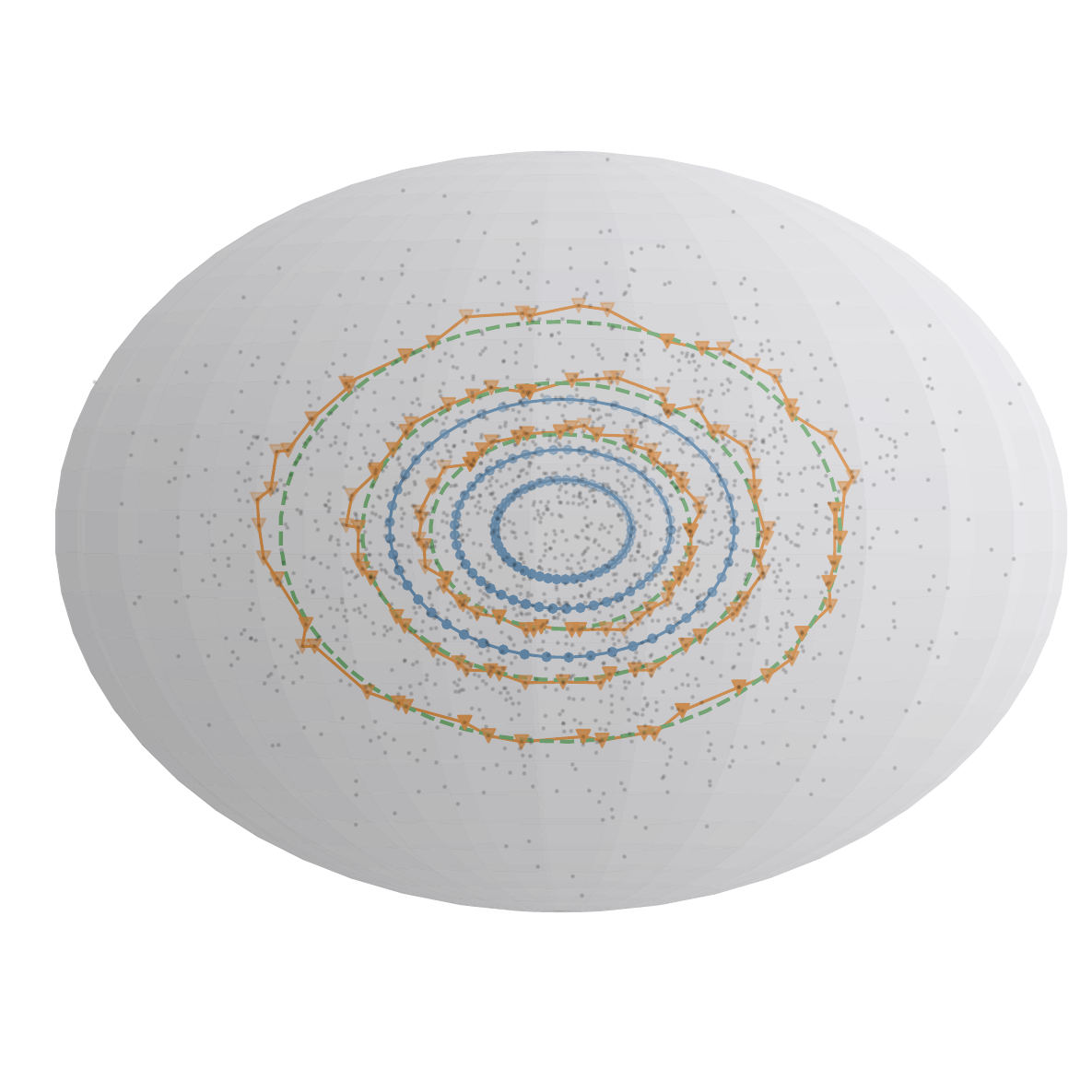}}}

\caption{Empirical regularized quantile contours (blue dots), unregularized ones (orange triangles), and ground truth (dashed green line).}
\label{fig_various_eps}
\end{figure}

\subsubsection{Structural differences between regularized and unregularized estimation}

The estimation of contours in \cite{hallin2022nonparametric} involves a matching between two samples of same size, 
hence it is not a functional object, in contrast to our entropic estimator.

\textbf{Two versus one OT problem.} Firstly, we stress that the estimation of contours in \cite{hallin2022nonparametric} requires to solve two different OT problems.
The construction of a uniform grid requires the choice of a reference central point (there is no origin on the sphere). 
Thus, a first OT problem estimates this central point $\FF(\theta_M)$ from the Fr\'echet median $\theta_M$, 
and a second one involves a grid (the reference contours) oriented towards the estimate of $\FF(\theta_M)$. 
This is not at all the case with our algorithm for which the estimate $\widehat{\uu}_{\ee,n}$ yields directly $\hat{\FF}_{N,n}^{\ee}$ and $\hat{\QQ}_{N,n}^{\ee}$, and a fortiori $\hat{\FF}_{N,n}^{\ee}(\theta_M)$.
Reference contours are then oriented to the central point $\hat{\FF}_{N,n}^{\ee}(\theta_M)$, and their image by $\hat{\QQ}_{N,n}^{\ee}$ yields quantile contours without the need to solve an additional OT problem. 
The time to compute a single OT/EOT problem is reported in Figure \ref{timecomput}, where it is shown that our algoritm is mostly profitable for large sample sizes.

\begin{figure}[htbp]
\centering
\includegraphics[width=0.4\textwidth,height=0.21\textwidth]{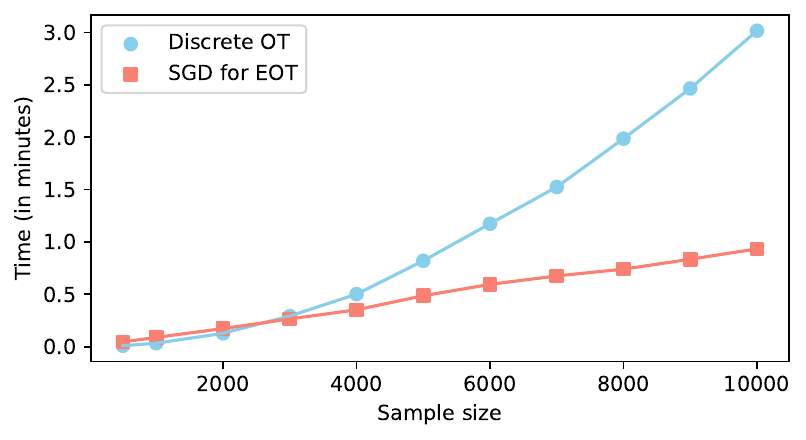}
\caption{Time for solving regularized and unregularized optimal transport between two data with $n$ instances, as a function of $n$.}
\label{timecomput}
\end{figure} 

\textbf{Out-of-sample estimation.}
Figure \ref{RegVsUnregMixture} illustrates MK quantile contours on a mixture of von-Mises Fisher distributions, with $\ee=0.1$. 
Because unregularized quantiles involve a bijection between the reference grid and the data, it is not able to compute the MK depth on points outside of the data. 
On the contrary, our entropic estimator $\hat{\FF}_{N,n}^{\ee}$ can be computed at any $x\in\mathbb{S}^{d-1}$, allowing to evaluate the MK depth at any out-of-sample observation. 
The upcoming numerical experiments illustrate how this specific property of our estimator is needed for descriptive analysis and depth-based classification.

\begin{figure}[htbp]
\centering
{\subfigure[Reference grid]{\includegraphics[width=0.3\textwidth]{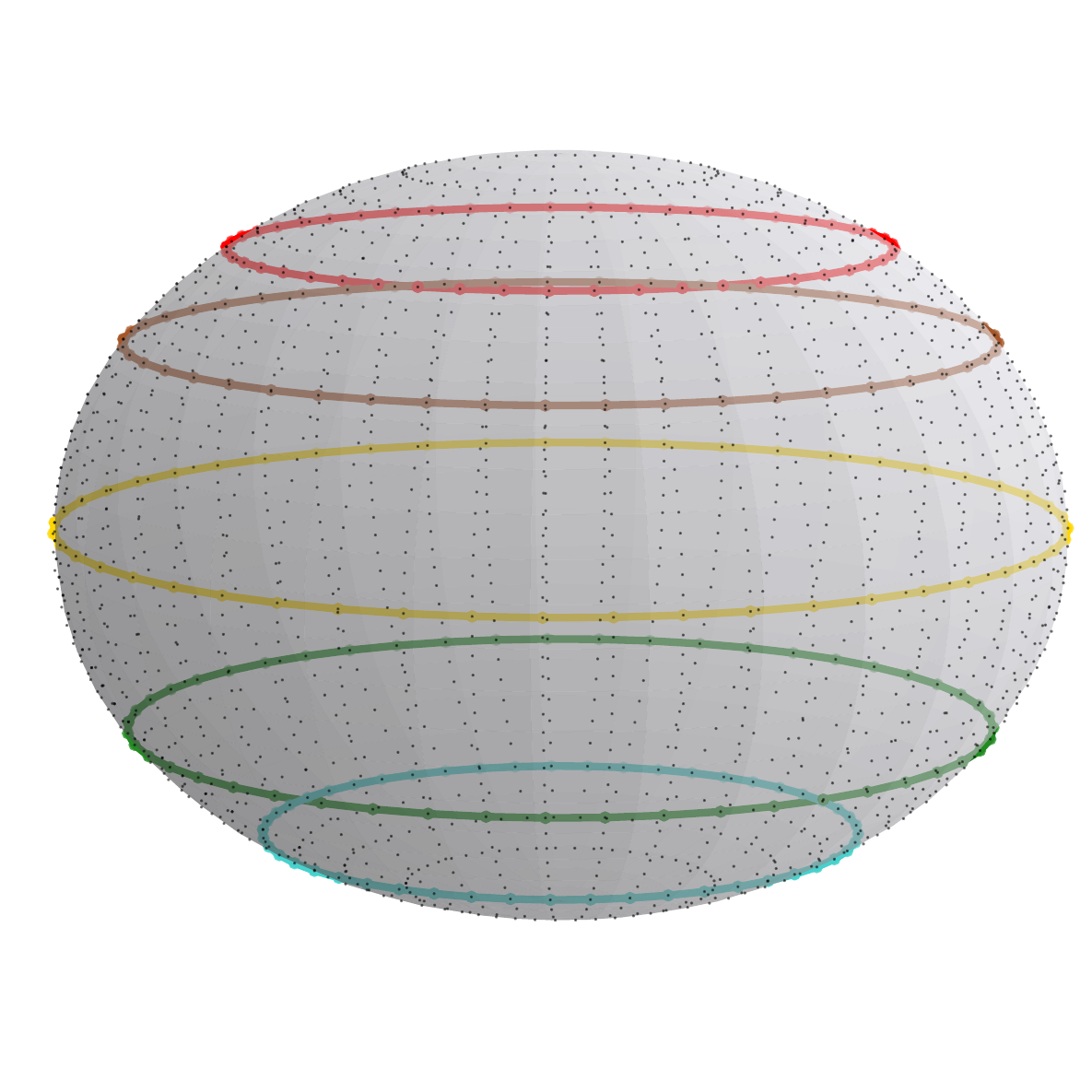}}}
{\subfigure[Unregularized contours]{\includegraphics[width=0.3\textwidth]{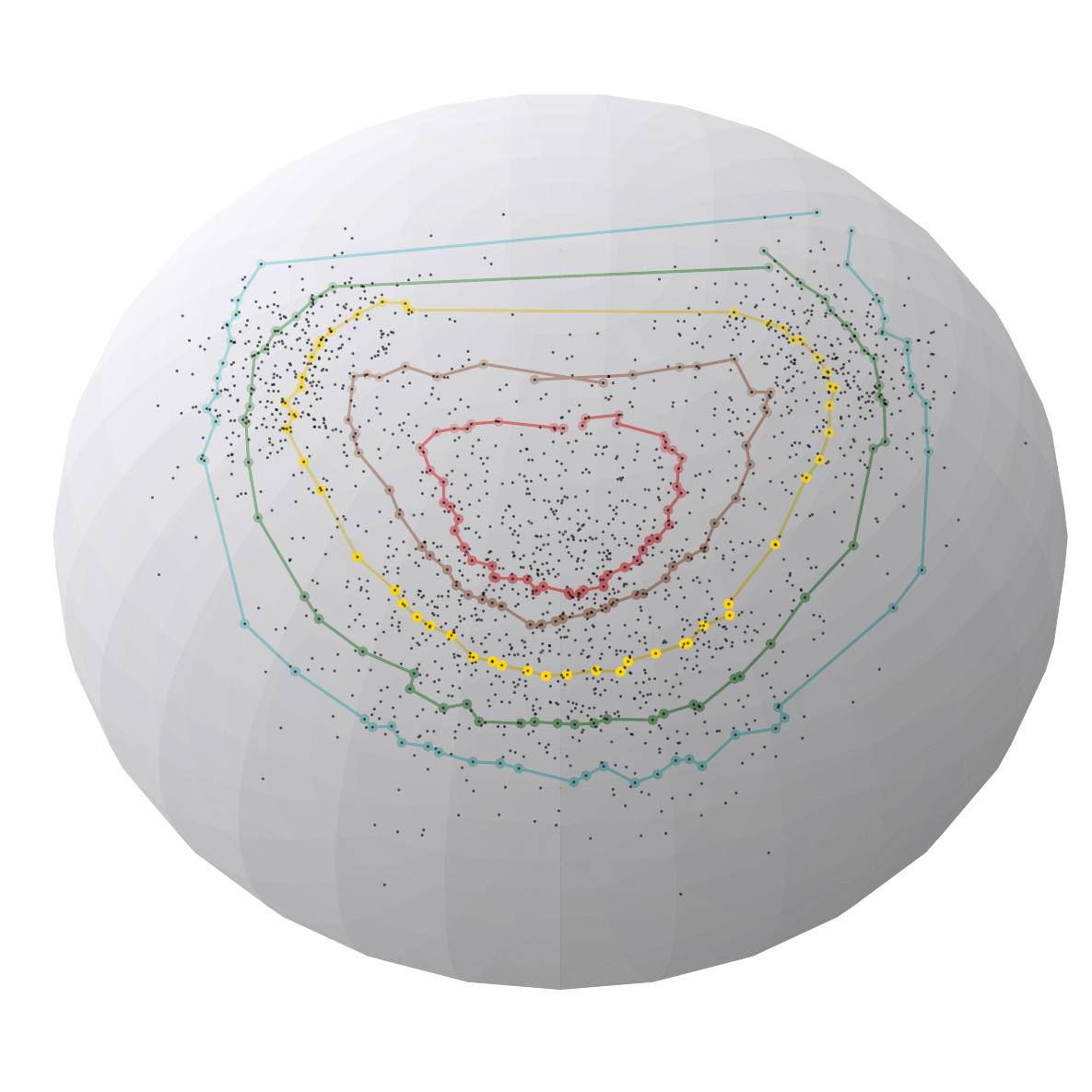}}}
{\subfigure[Regularized contours $\epsilon=0.1$]{\includegraphics[width=0.3\textwidth]{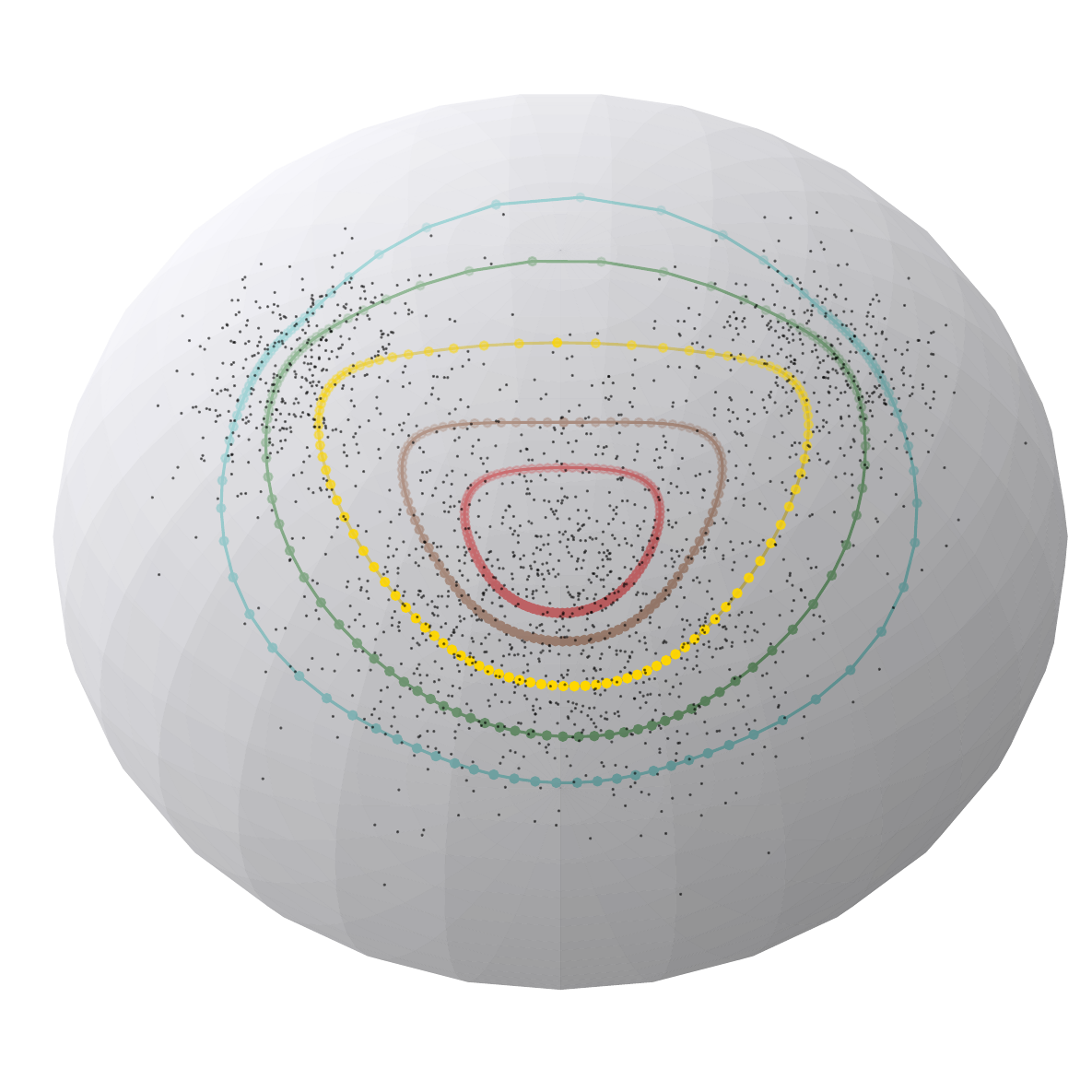}}}

\caption{Empirical MK quantile contours on a mixture of vMF distributions}
\label{RegVsUnregMixture}
\end{figure}

\subsection{Descriptive tools}\label{descriptiveTools}

The seminal paper \cite{liu1999multivariate} gathers descriptive tools based on data depths.
Monge-Kantorovich analogs already exist for data in $\RR^d$, and we shall now extend some of them to the directional setting.
We stress that the ability of our regularized estimator to interpolate between data points is crucial for $(i)$ smooth contours in practice and $(ii)$ computing volumes of quantile regions.

\subsubsection{Representative plots}

Firstly, \cite{liu1999multivariate} study representative plots for bivariate data, and 
the MK analog is given by the descriptive plots from \cite{Hallin-AOS_2021,hallin2022nonparametric}, the latter with the added information of sign curves. 
Figure \ref{figSC} illustrates it on a Tangent von-Mises Fisher distribution, \cite{garcia2020optimal}, and on a Mixture of two von-Mises Fisher distributions, with the help of our empirical regularized quantile function. 
One can observe that
the shapes of the distributions are well recovered. 
Because entropic MK quantiles interpolate between data points, contours cross the void between mixture components on the right-hand side. 
A careful inspection shows that the number of points per contour within this void is much lower than in the high density areas. 
This illustrates how the variation of mass, that is the underlying geometry, is captured by our regularized estimator. 

\begin{figure}[htbp]
\centering
\includegraphics[width=0.35\textwidth]{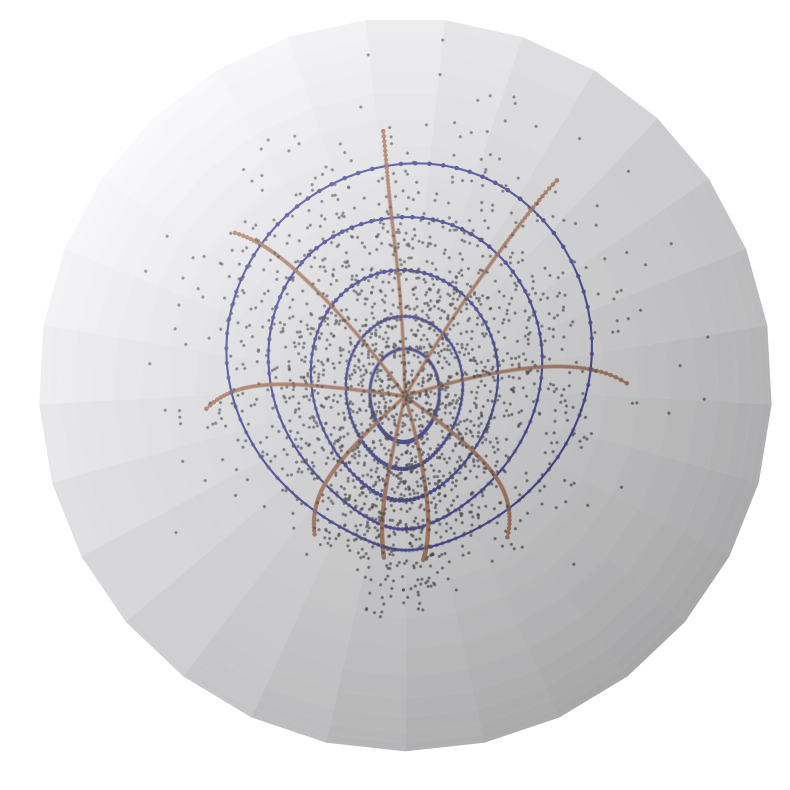}
\includegraphics[width=0.35\textwidth]{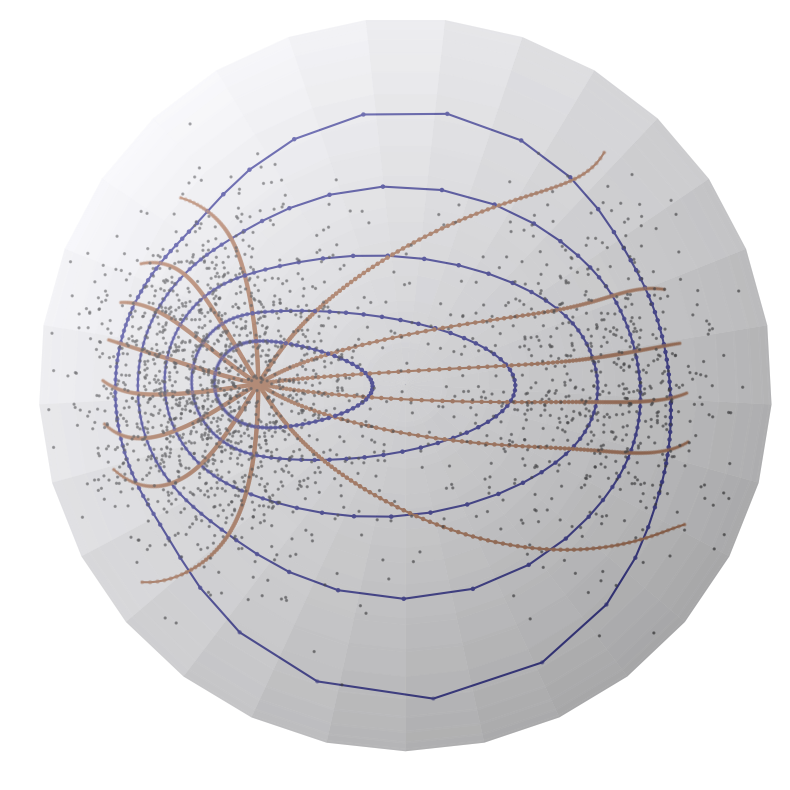}

\caption{Regularized quantile contours of levels $\{0.1,0.25,0.5,0.75,0.9 \}$ and associated sign curves, with $\epsilon= 0.1$. }
\label{figSC}
\end{figure}

\subsubsection{Scale or dispersion}

We now present a graphical tool to describe the amount of dispersion, called the \textit{scale curve} in \cite{liu1999multivariate} and whose MK analog has been introduced in \cite{beirlant2019centeroutward} for Euclidean data. 
The \textit{scale curve} is a plot of the volumes $V(\alpha)$ of MK quantile regions with respect to $\alpha \in[0,1]$.
The faster it grows, the greater the dispersion. Thus, if the scale curve of $\nu_1$ is consistently above the one of $\nu_2$, then $\nu_1$ is more spread out than $\nu_2$.
Define 
$$
V(\alpha) = \int_{\mathbb{C}_\alpha} d\mu_\Sph(x) 
= \int_\Sph \mathds{1}_{\{x \in \mathbb{C}_\alpha\}} d\mu_\Sph(x)
= \int_\Sph \mathds{1}_{\{ \langle \FF(x) ,\FF(\theta_M) \rangle \geq 1 - 2\alpha\}} d\mu_\Sph(x).
$$
This can be estimated with a sample $U_1,\cdots, U_N$ from $\mu_\Sph$, by the proportion
$
V_{\ee,n}(\alpha) = \frac{1}{N} \sum_{i=1}^N \mathds{1}_{\{ \langle \hat{\FF}_{N,n}^{\ee}(U_i) ,\FF(\theta_M) \rangle \geq 1 - 2\alpha\}}.
$
Crucially, this requires to estimate $\FF(x)$ for out-of-sample observations $x$, for which regularization is needed.
On the left-hand side of Figure \ref{scalecurves}, we draw the scale curves of von-Mises Fisher distributions with varying concentration parameter $\kappa \in \{1,2,5,15\}$, which controls the dispersion of samples. 
It is well-captured that the lower the value of $\kappa$, the more spread out is the underlying distribution.
We also illustrate in the right-hand side of Figure \ref{scalecurves} that the scale curve captures the presence of outliers.
We consider three identical vMF distributions with dispersion parameter $\kappa = 15$ and mean $(0,1,0)^T$. 
Adding $N \in \{5,20,50\}$ outliers near from the North Pole $(0,0,1)^T$,
 the dispersion increases with 
the number of outliers,
which is precisely the expected behavior.

\begin{figure}[htbp]
\centering

{\subfigure{\includegraphics[width=0.3\textwidth,height=0.29\textwidth]{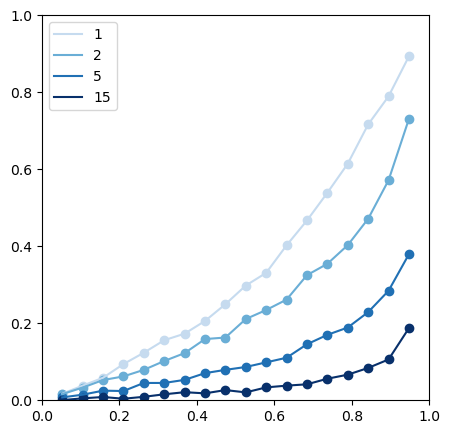}}}
\hspace{.5cm}
{\subfigure{\includegraphics[width=0.3\textwidth,height=0.29\textwidth]{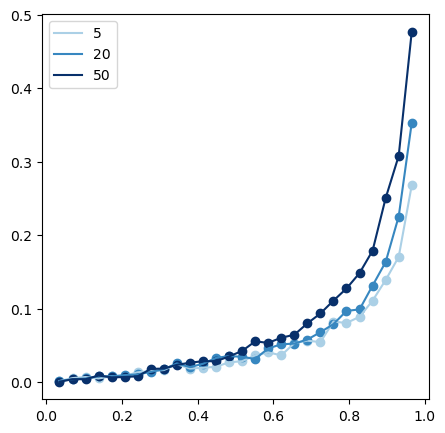}}}

\caption{Scale curves of von-Mises Fisher distributions with various (left) dispersion parameter $\kappa$ and (right) number $N$ of added outliers.}
\label{scalecurves}
\end{figure}

\subsection{Other concepts of quantiles}\label{numexp_compQuantiles}

In Figure \ref{fig_various_ee}, we display a visual comparison of existing concepts for quantiles on the sphere by focusing on mixtures of von-Mises Fisher distributions.
All distributions are centered for ease of visualization with azimuth angle $90$ degrees and elevation angle $5$ degrees.
 It can be observed from  Figure \ref {fig_various_ee} that
Mahalanobis quantile regions  \cite{ley2014new} are concentric spherical caps, whereas spatial quantiles \cite{konen2023spatial} and our regularized MK quantiles can exhibit more complex shapes that better fit the geometry of the data.
To that extent, spatial and regularized MK quantiles,  
are  both more satisfactory.  
For each notion of quantiles, $100$ points are drawn within each contours, with straight lines to link them.
We emphasize that spatial quantiles are not indexed by their probability content, as opposed to the MK ones \cite{hallin2022nonparametric}.

\begin{figure}[htbp]
\centering
{\subfigure[Mahalanobis ]{\includegraphics[width=0.31\textwidth]{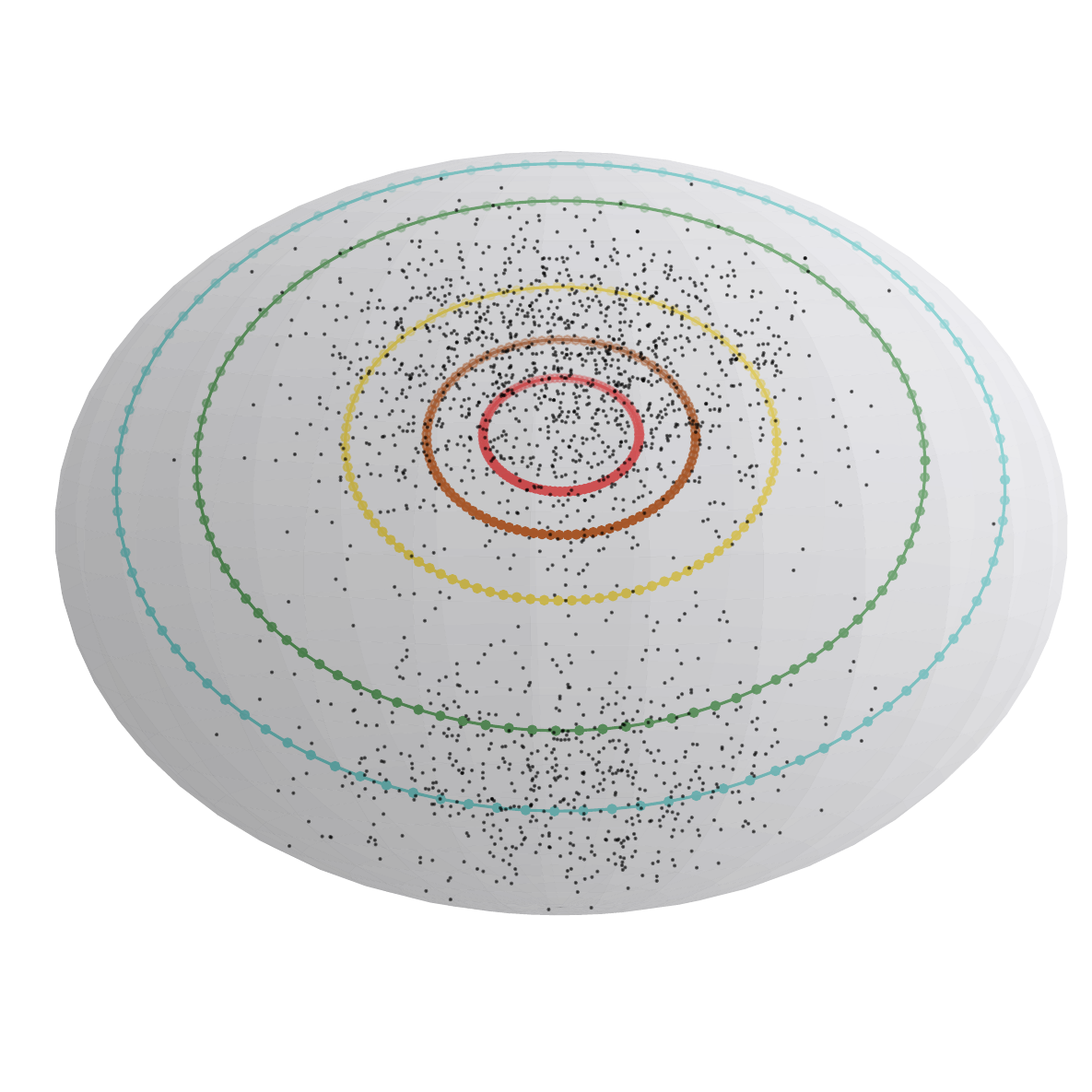}}}
{\subfigure[Spatial ]{\includegraphics[width=0.31\textwidth]{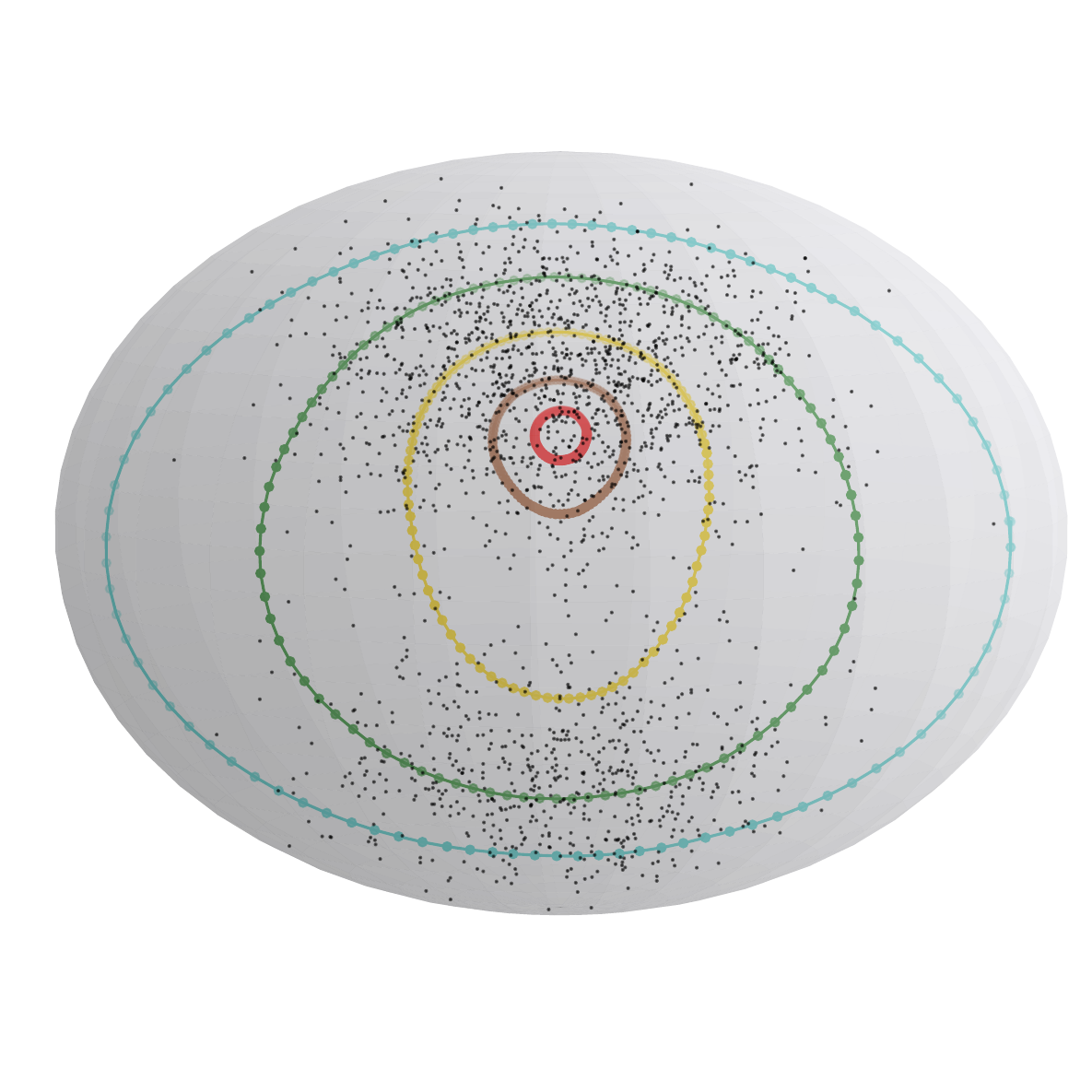}}}
{\subfigure[Monge-Kantorovich ]{\includegraphics[width=0.31\textwidth]{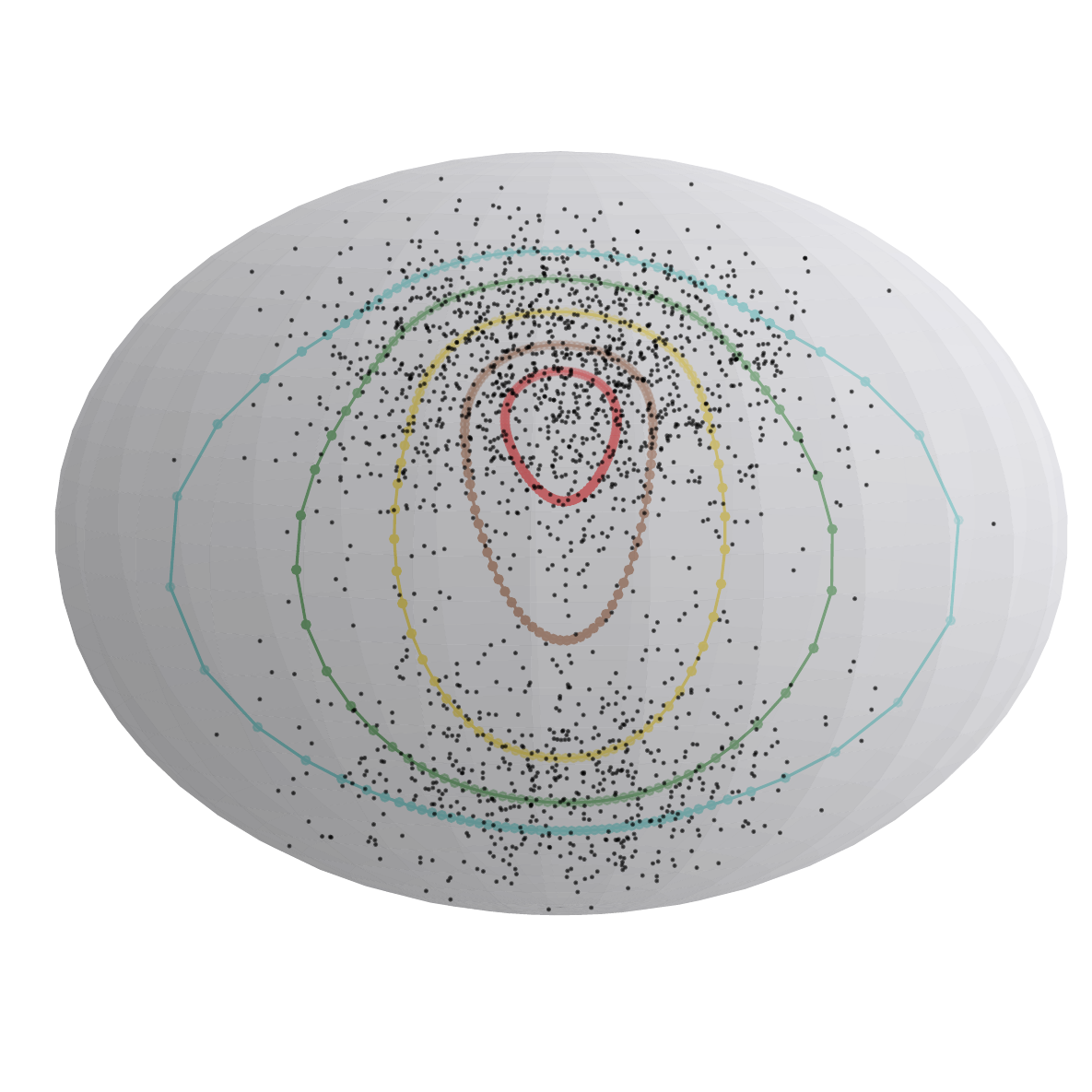}}}

{\subfigure[Mahalanobis ]{\includegraphics[width=0.31\textwidth]{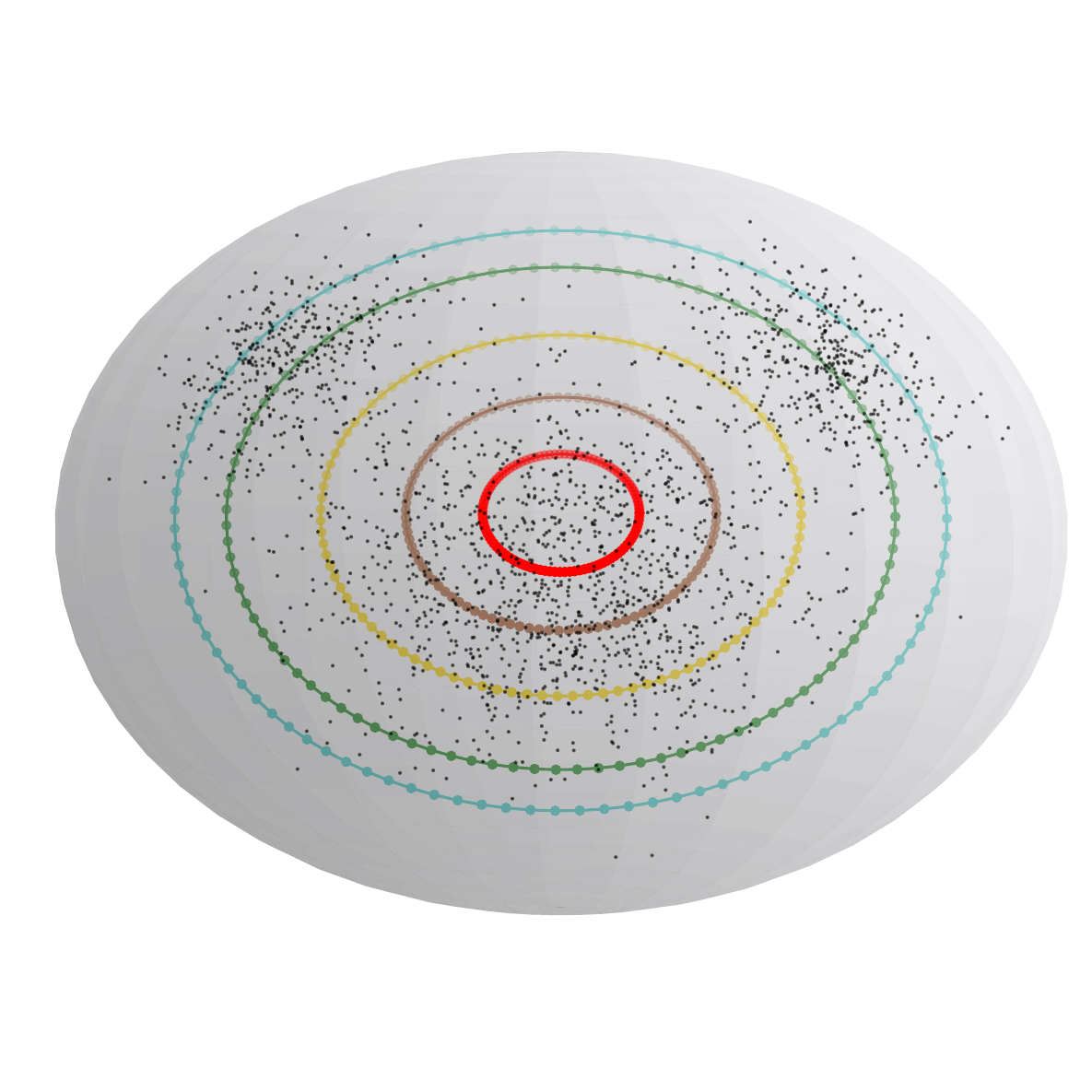}}}
{\subfigure[Spatial ]{\includegraphics[width=0.31\textwidth]{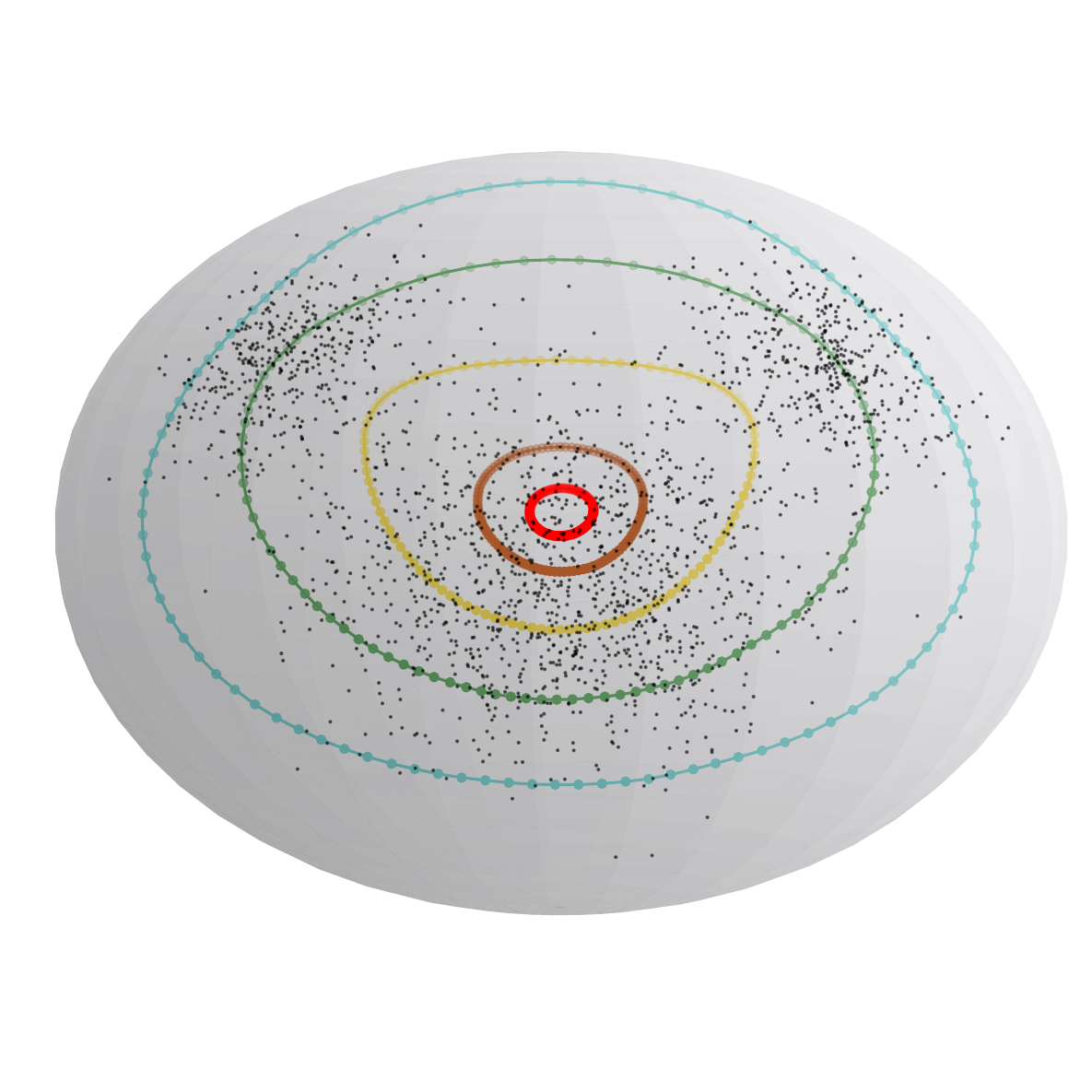}}}
{\subfigure[Monge-Kantorovich]{\includegraphics[width=0.31\textwidth]{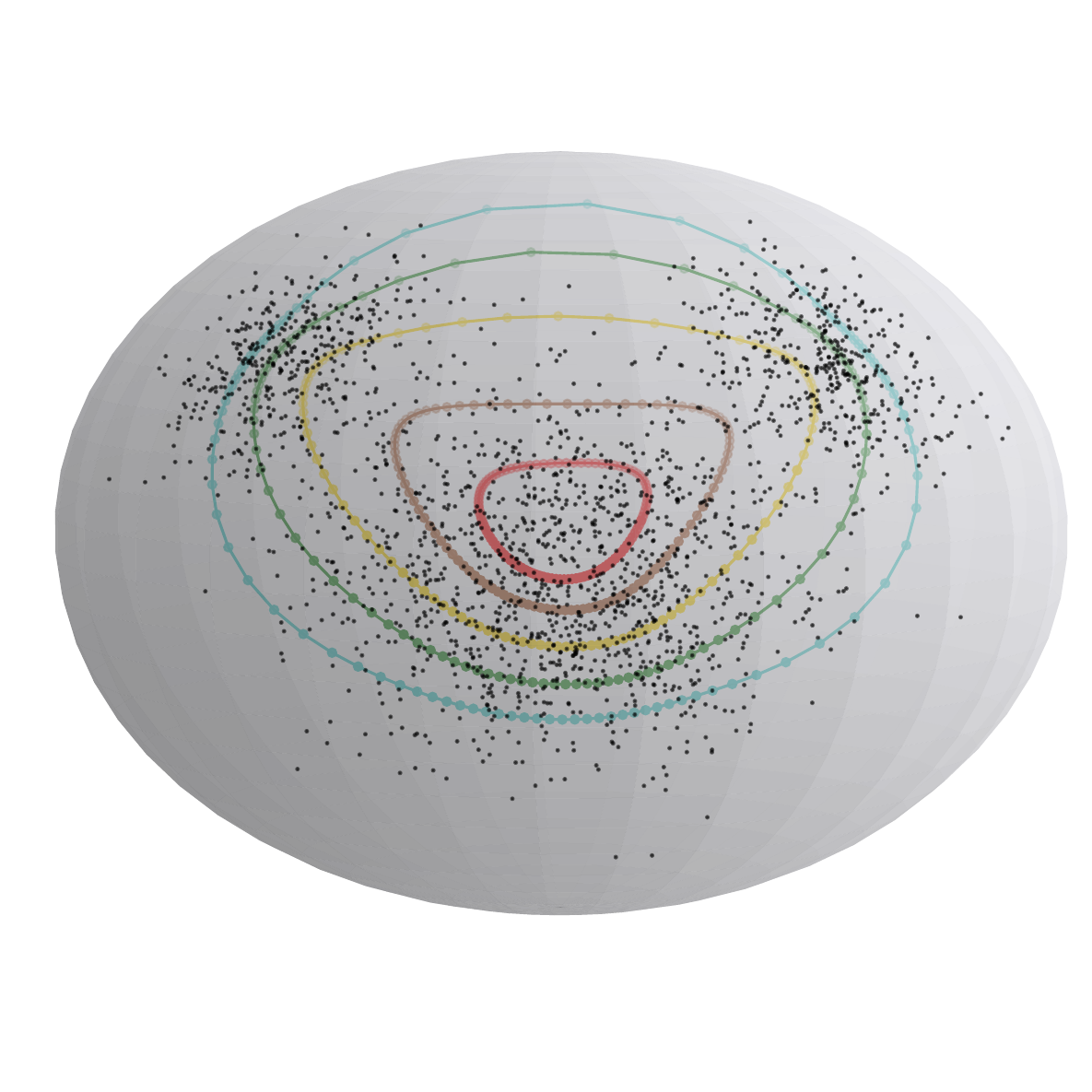}}}

{\subfigure[Mahalanobis]{\includegraphics[width=0.31\textwidth]{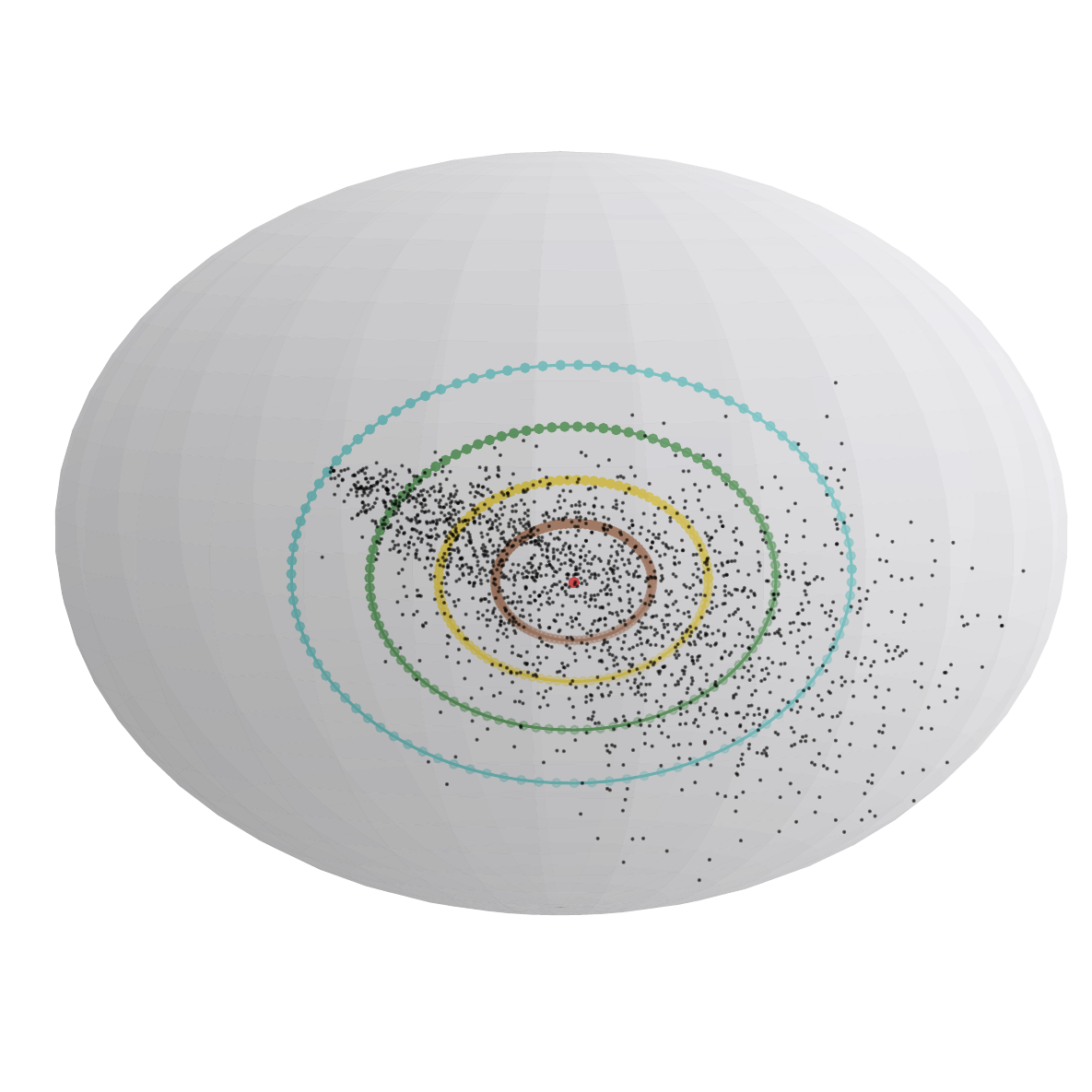}}}
{\subfigure[Spatial ]{\includegraphics[width=0.31\textwidth]{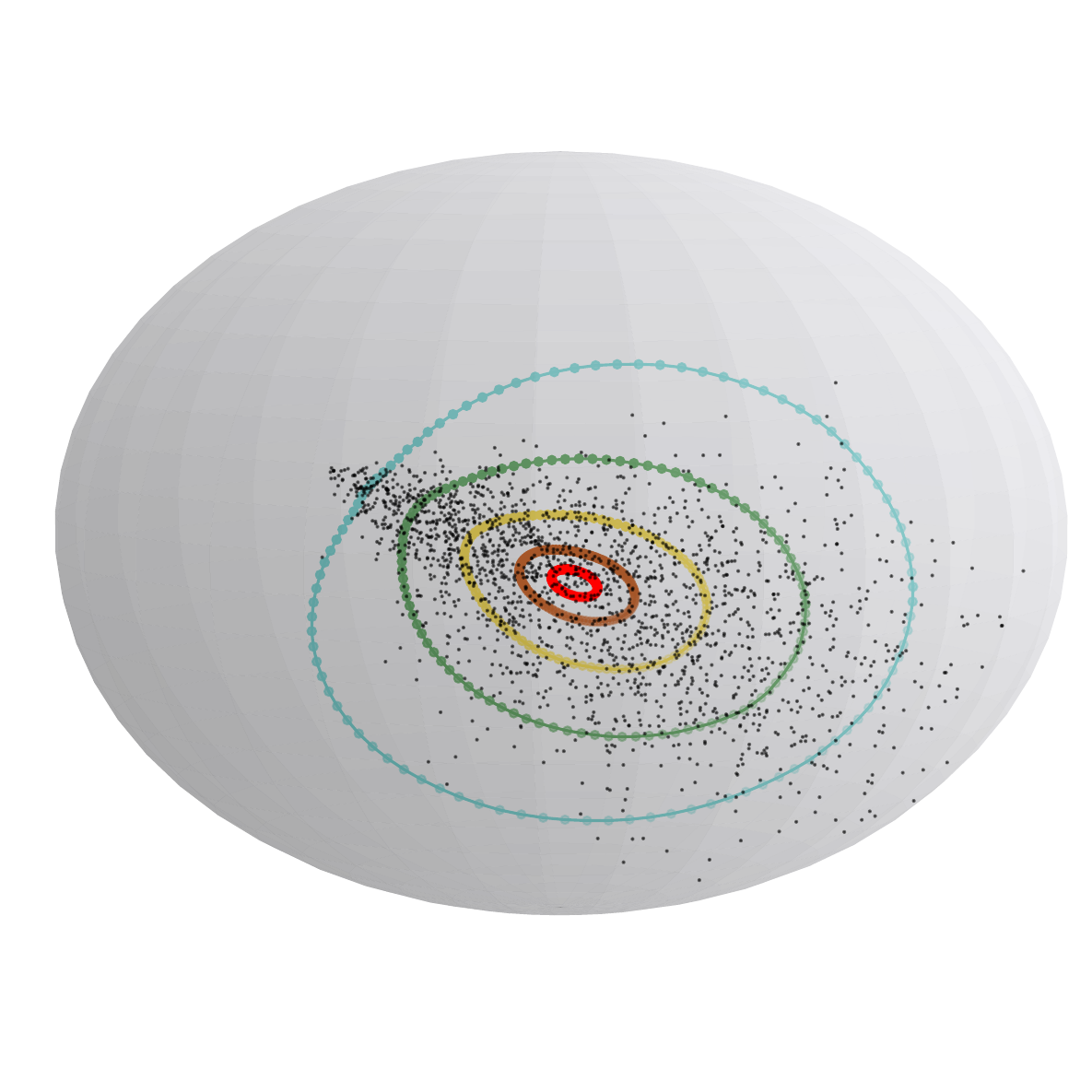}}}
{\subfigure[Monge-Kantorovich]{\includegraphics[width=0.31\textwidth]{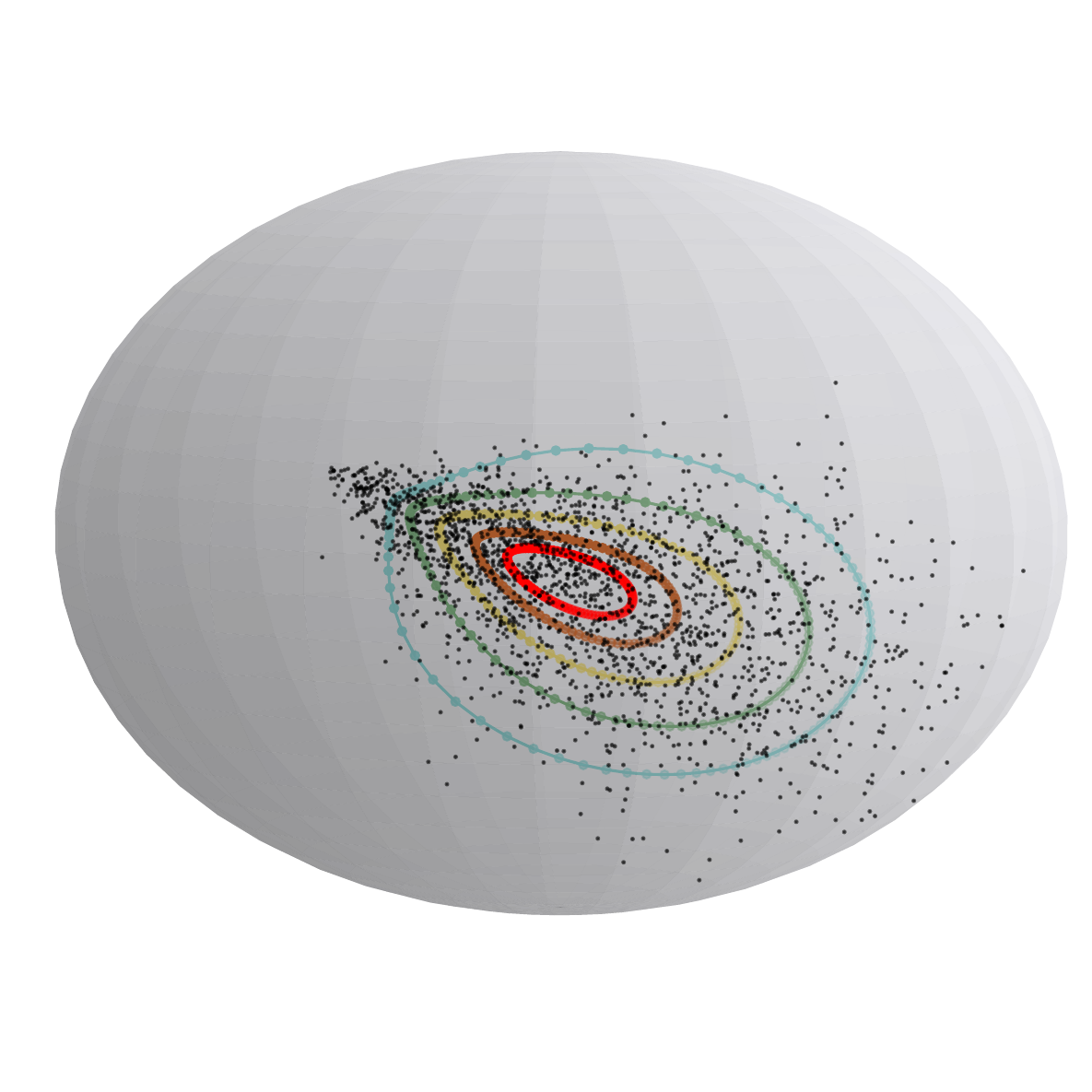}}}

\caption{Quantile contours $\{0.1,0.25,0.5,0.75,0.9\}$ for different definitions of quantiles and data sampled from mixtures of two (first row) and three (second row)  von-Mises Fisher distributions, and from a Tangent vMF (third row).}
\label{fig_various_ee}
\end{figure}

\subsection{Max-depth classification}\label{maxdepthclassif}

Following experiments in \cite{konen2023spatial}, we report a \textit{quantitative comparison} between different notions of directional quantiles on the task of supervised classification, based on the \textit{max-depth} approach from \cite{ghosh2005maximum}. 
Because statistical depths are designed to measure \textit{outlyingness} within a distribution, one can classify $x\in\Sph$ as coming from the distribution $\nu_1$ instead of $\nu_2$ if the depth of $x$ with respect to ($w.r.t.$) $\nu_1$ is greater than the one $w.r.t.$ $\nu_2$. 
As highlighted in \cite{konen2023spatial}, the vanishing property of many statistical depths \cite{liu1992ordering,tukey75}, raises issues when the sample $x$ lies outside the convex hull of empirical data, where new data is classified randomly.
On the contrary, the depth associated to spatial quantiles \cite{konen2023spatial}[Section 5], is strictly positive everywhere, and the same holds trivially for the MK depth.
The regularization parameter for EOT is taken as $\ee=10^{-1}$.
 Define $O$ the rotation matrix fixing the first coordinate and mapping $m_1 =(0,0,1)^T$ to $m_2=(0,\sin(\pi/3),\cos(\pi/3))^T$. 
 We consider for $(P_1,P_2)$ : 
\begin{enumerate}
\item \textit{von-Mises Fisher:} $P_1$ is a vMF with location $m_1$ and concentration $\kappa=3$, $P_2$ is the distribution of $OX$ for $X \sim P_1$.
\item \textit{Tangent von-Mises Fisher:} $P_1$ is the distribution of 
$$
X = Zm_1 + \sqrt{1-Z^2} S,
$$
where $Z = 2V-1$ for $V$ a $\text{Beta}(2,8)$ distribution, $S_3=0$ and $(S_1,S_2)^T$ follows a von-Mises Fisher distribution with location $(1,0)^T$ and concentration $\kappa'=5$. 
Again, $P_2$ is the distribution of $OX$ for $X \sim P_1$.
\item \textit{Non-convex data :} $P_1$ is the same mixture of three vMF distributions than in Figure \ref{fig_various_ee}. $P_2$ is a vMF distribution centered at $(0,0,1)^T$ with concentration parameter $\kappa = 50$. 
\end{enumerate}

Examples of data sampled from these three distributions are displayed  in Figure \ref{settings}, with observations sampled from $P_1$ in blue and $P_2$ in red.

\begin{figure}[htbp]
\centering
{\subfigure[von-Mises Fisher]{\includegraphics[width=0.3\textwidth]{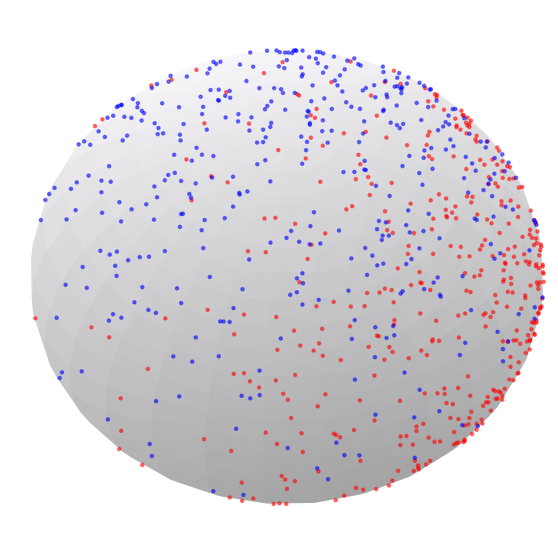}}}
{\subfigure[Tangent von-Mises Fisher]{\includegraphics[width=0.3\textwidth]{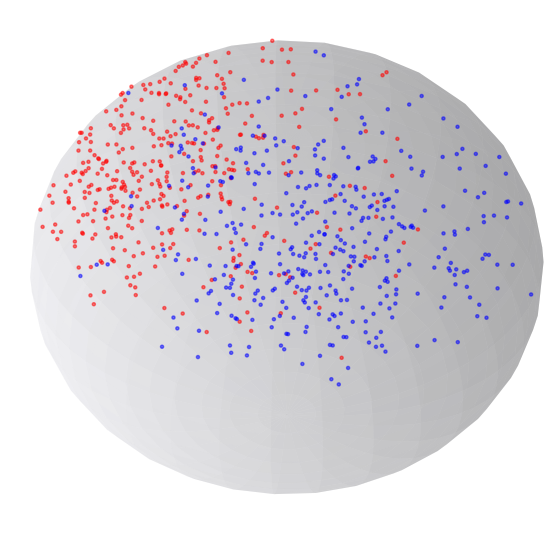}}}
{\subfigure[Non-convex data]{\includegraphics[width=0.3\textwidth]{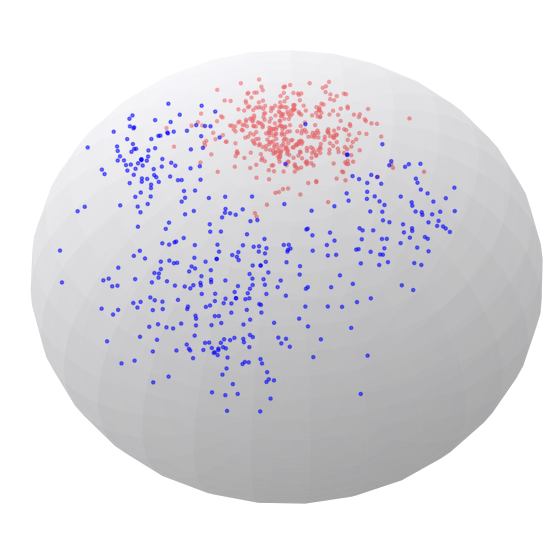}}}

\caption{Data for max-depth classification.}
\label{settings}
\end{figure}

In each setting, we learn the quantile functions on $200$ points coming from distributions $(P_1,P_2)$, before computing misclassification rates on a test data of $200$ different observations. 
This is repeated $100$ times and the misclassification rates are gathered in  
Figure \ref{fig_missclass}.
All methods provide comparative results when both the supports of $P_1$ and $P_2$ are convex. However, in the third setting, classification based on the MK depth significatively outperforms the others. 
This illustrates the potential of entropically regularized quantiles for statistical purposes.

\begin{figure}[h!]
\centering
{\subfigure[Von-Mises Fisher]{\includegraphics[width=0.32\textwidth]{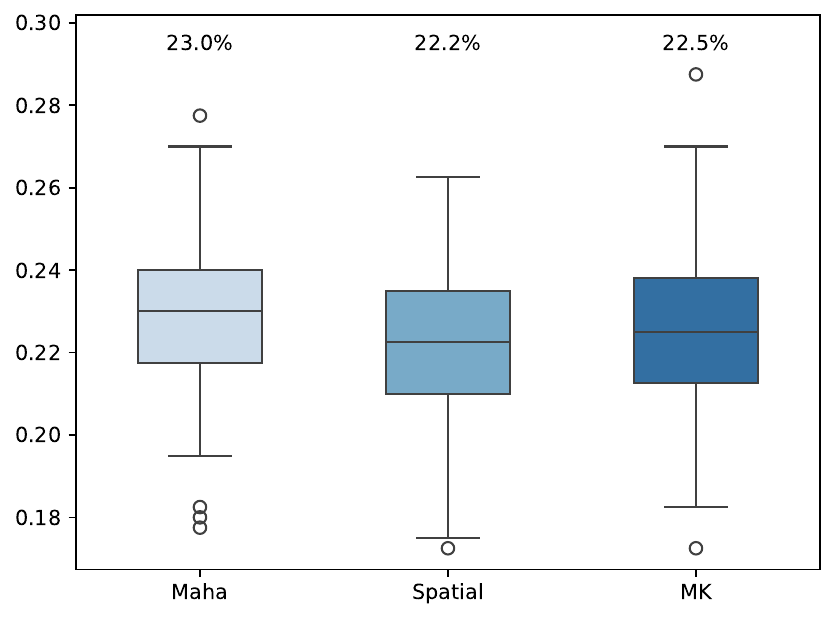}}}
{\subfigure[Tangent von-Mises Fisher]{\includegraphics[width=0.32\textwidth]{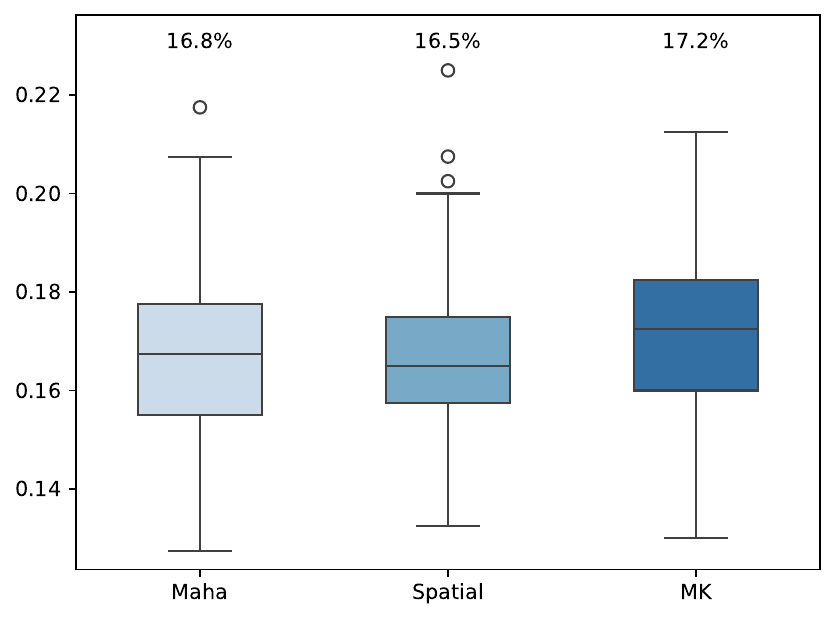}}}
{\subfigure[Non-convex data]{\includegraphics[width=0.32\textwidth]{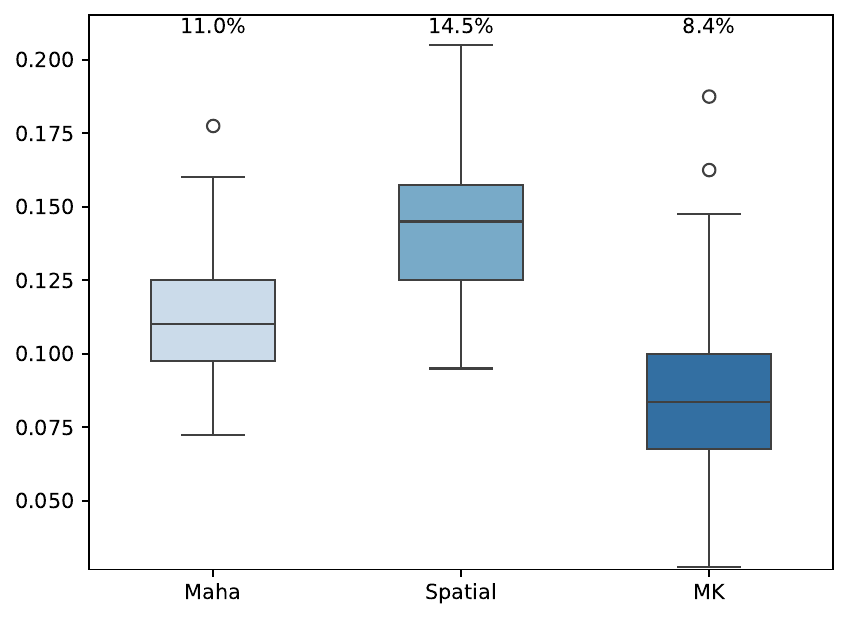}}}

\caption{Misclassification rates with several statistical depths.}
\label{fig_missclass}
\end{figure} 

These results are complemented in Figure \ref{fig_missclass_variouseps} with misclassification rates of the MK depth in the non-convex setting, but changing the regularization parameter. 
This highlights that other values of $\epsilon$ lead to even better classification results.
Thus, beyond the convenient baseline of $\epsilon=0.1$ that performs well overall, one may consider cross-validation of $\epsilon$ if enough data is available.
Figure \ref{fig_missclass_variouseps} also contains further comparison between depth notions in the non-convex setting and with $\epsilon = 0.1$.
There, the training and testing subsets have a larger sample size, either $n=500$ or $n=1000$. This validates results from Figure \ref{fig_missclass}.

\begin{figure}[htbp]
\centering
{\subfigure[Changing the regularization $\epsilon$]{\includegraphics[width=0.35\textwidth]{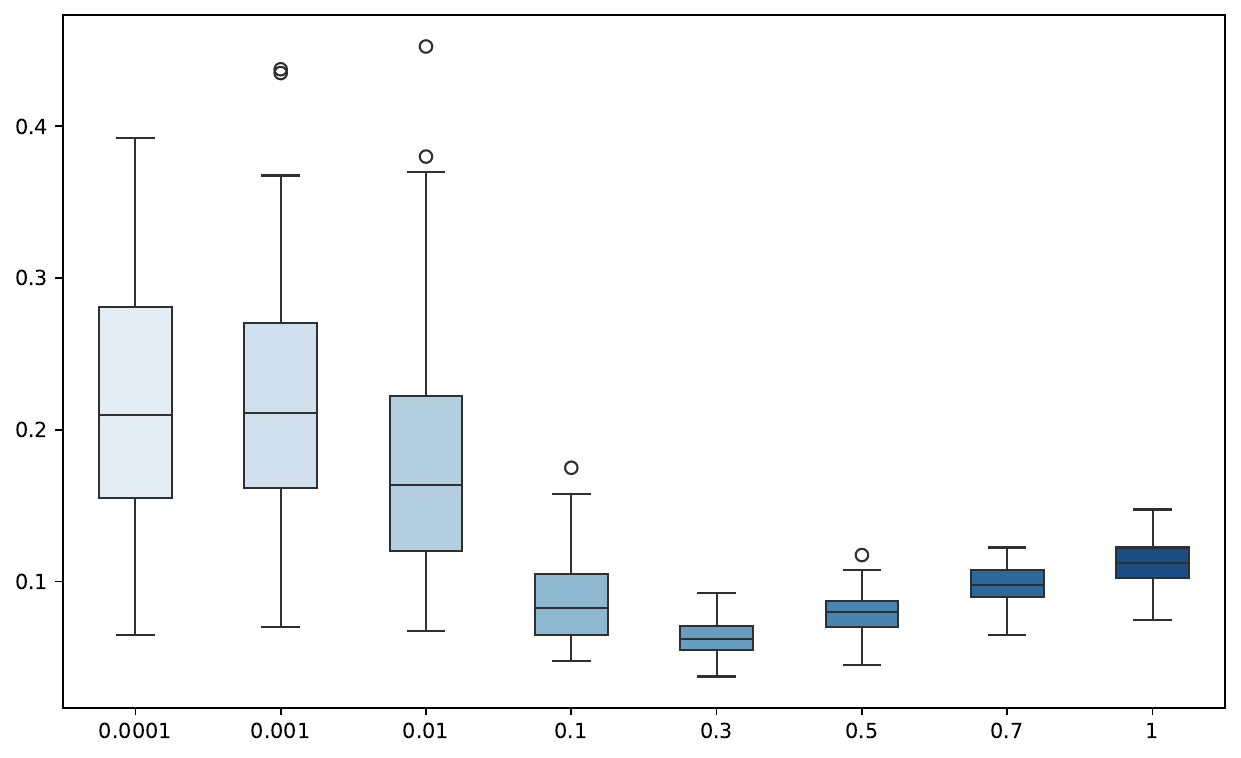}}}
{\subfigure[$n=500$]{\includegraphics[width=0.29\textwidth]{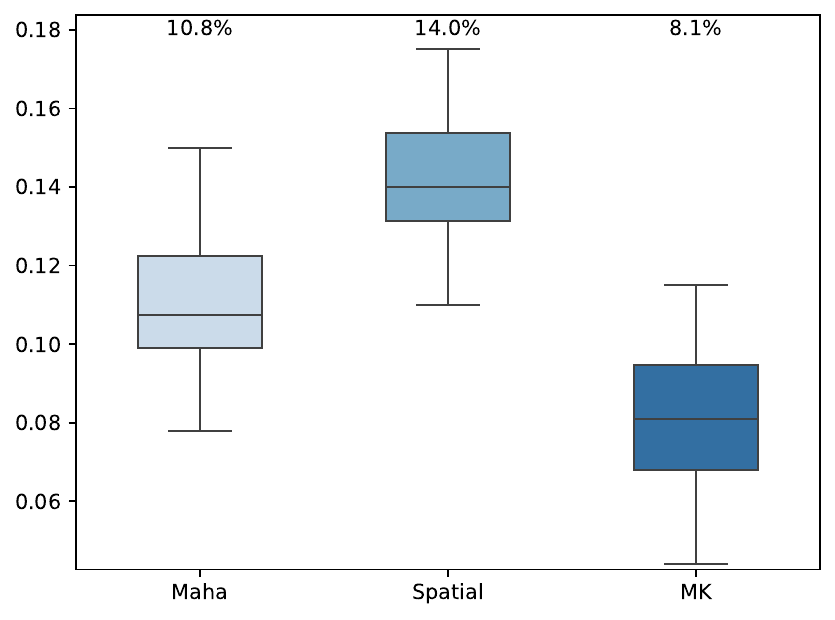}}}
{\subfigure[$n=1000$]{\includegraphics[width=0.29\textwidth]{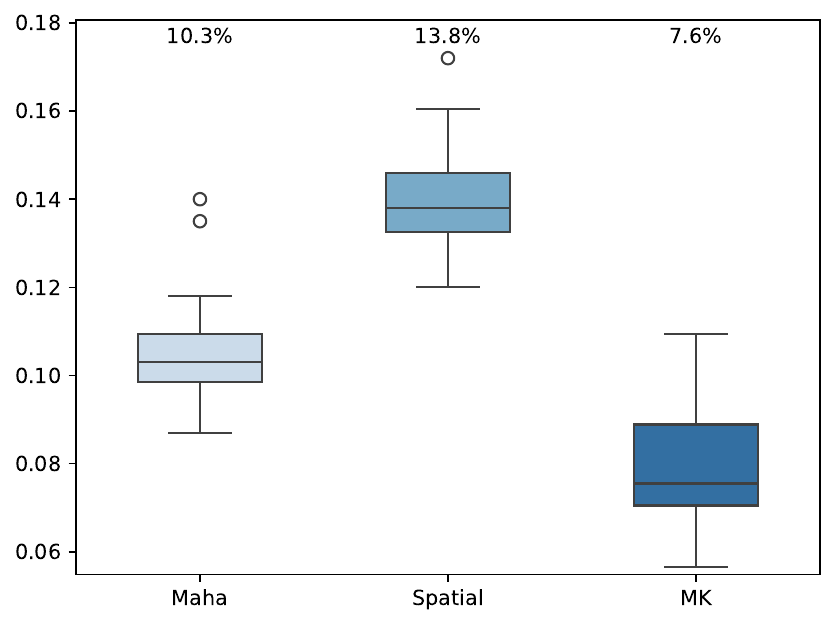}}}

\caption{Further misclassification rates in the non-convex setting, with various $\epsilon$ values and for larger sample size.}
\label{fig_missclass_variouseps}
\end{figure}

\section{Concluding remarks}  \label{sec:diss}

The major limitation of discrete OT for quantiles estimation is the impossibility to provide out-of-sample estimates $Q(x)$ that interpolate the data.
In the present paper, we showed that regularizing by entropy can be used as an alternative, particularly when the focus is on the quantile contours or the volumes of quantile regions. 
Still, we emphasize that this regularization loses the distribution-freeness of the ranks, which is crucial for testing purposes, \cite{hallin2022nonparametric}. Thus, the choice of using OT or EOT shall depend on the considered task.

Our numerical scheme leverages spherical Fourier series to construct a stochastic gradient descent to solve \textit{continuous} OT in the limit of the iterations. 
This is particularly useful when the number of observations $n$ is large and/or when data arrives sequentially.
A natural perspective would be to extend these results to other manifolds, in light with the recent work \cite{hallin2024quantilesManifold}, or to the statistics of circular data, discussed in detail in \cite{Hundrieser2022}. 
The main limitation of the present work concerns the dimensionality, 
even though most applications in directional statistics has been implemented in dimension $d=3$. 
This is inherited from the computation of spherical harmonics, since we use the Python library \textit{pyshtools} \cite{wieczorek2018shtools}. 
However, our implementation could in principle be extended to $d>3$ from other spherical harmonics transforms.

\appendix 
\section*{Appendix}
\renewcommand\thesection{\Alph{section}}

The supplementary material gathers detailed proofs for all the results related to the differentiation of entropic maps and potentials. 
It also describes how the explicit forms of spherical MK quantiles from \cite{hallin2022nonparametric} are computed in our numerical experiments, as they simplify in dimension $d=3$.

\section{Explicit forms}\label{explicitforms}

Closed-form expressions of $\FF$ for rotationally invariant distributions were given in \cite{hallin2022nonparametric} allowing to deduce the inverse map $\QQ$. 
Because part of our numerical experiments rely on these, this section gathers such explicit forms, as they simplify in dimension $d=3$. 

Let $Z\sim \nu$ be such a random vector with axis $\pm \theta_M$. Then, assume that $\nu$ has density 
$$
z \in \Sph \mapsto c_f f( z^T\theta_M ),
$$
for $f$ some positive \textit{angular function} and $c_f$ a normalizing constant.
For $r \in [-1,1]$, denote by
$$
F_f(r) = \int_{-1}^r f(s) ds / \int_{-1}^1 f(s) ds
$$
the distribution function of $Z^T\theta_M$ and by $Q_f = F_f^{-1}$ its quantile function.
Then, letting $F_f^*(r) = 2F_f(r) -1$, the directional distribution function of $Z$ writes 
\begin{equation}\label{explicitFFrotinv}
\FF(z) = F_f^*(z^T\theta_M) \theta_M + \sqrt{1 - F_f^*(z^T\theta_M)^2} S_{\theta_M}(z).
\end{equation}
For instance, taking $f(s) = \exp(\kappa z^T\theta_M)$ corresponds to the von Mises-Fisher distribution with location parameter $\theta_M$ and concentration parameter $\kappa \in \RR_+$.
Crucially, the transport \eqref{explicitFFrotinv} reduces to univariate transport along the axis $\theta_M = \FF(\theta_M)$.
If $\theta_M = (0,0,1)^T$, that is to say \textit{up to some rotation} thanks to equivariance properties shown in \cite{hallin2022nonparametric}, this corresponds to changing the latitude $w.r.t.$ the usual coordinate system
\begin{equation}\label{Phispheric_coord}
x = \Phi(\theta,\phi) := (\cos \phi \sin \theta, \sin \phi\sin\theta,\cos\theta).
\end{equation}
Indeed, as soon as $\theta_M = (0,0,1)^T$, 
\begin{equation*}
\FF(z) = F_f^*(z_3) \theta_M + \sqrt{1 - F_f^*(z_3)^2} \frac{(z_1,z_2,0)^T}{\Vert (z_1,z_2,0)^T \Vert}.
\end{equation*}
The third coordinate is changed to $F_f^*(z_3)$, and the other coordinates are adapted to the constraint $\FF(z)\in\Sph$.
This rewrites, in accordance with \eqref{Phispheric_coord},
\begin{equation}\label{explicitFFrotinv2}
\FF(z) = \FF(\Phi(\theta,\phi))= \Phi(\overline{\theta},\phi)
\hspace{1cm} \mbox{for} \hspace{1cm} 
\overline{\theta}= \arccos\Big((F_f^*)(z_3)\Big).
\end{equation}
Consequently, to get the inverse map $\QQ = \FF^{-1}$, it suffices to change the pseudo latitude of $\FF(z) \in \mathcal{C}_\tau^U$ $w.r.t.$ the axis $\pm \theta_M$. 
If $x = \FF(z)$, $x_3 = F_f^*(z_3)$ and $z_3 = (F_f^*)^{-1}(x_3)$, that is
\begin{equation}\label{explicitQQrotinv}
\QQ(x) = \QQ(\Phi(\theta,\phi)) = \Phi(\tilde{\theta},\phi) 
\hspace{1cm} \mbox{for} \hspace{1cm} 
\tilde{\theta}= \arccos\Big((F_f^*)^{-1}(x_3)\Big).
\end{equation}
As highlighted in \cite{hallin2022nonparametric}, this shows that MK quantile contours coincide with Mahalanobis ones from \cite{ley2014new} under the rotationally symmetric model.

\section{First-order differentiation of entropic potentials}

This section is dedicated to the proof of Proposition \ref{entropicQeeFee}, that gives explicit spherical entropic maps.

\subsection{Euclidean derivatives} 

\begin{prop}\label{entropicmaps}
Denote by 
\begin{equation}\label{def_ge}
g_\ee(x,z) = \frac{-d(x,z)}{\sqrt{1 - \langle x,z\rangle^2}} \exp \Big( \frac{\uu_\ee(x) - c(x,z) + \uu_\ee^{c,\ee}(z)}{\ee} \Big).
\end{equation}
Then, the Euclidean partial derivatives of $\uu_\ee$ admit the closed-form expression
\begin{equation}\label{grad_u_ee}
\partial_{x_i} \uu_\ee(x) = \int z_i g_\ee(x,z) d\nu(z).
\end{equation}
Similarly, the Euclidean gradient of $\uu_\ee^{c,\ee}$ verifies
\begin{equation}\label{grad_u_ee_c_ee}
\partial_{z_i} \uu_\ee^{c,\ee}(z) = \int x_i g_\ee(x,z) d\mu_\Sph(x).
\end{equation} 
\end{prop}

Denoting by $v = \uu_\ee^{c,\ee}$, rewriting the optimality condition $\uu_\ee = (\uu_\ee^{c,\ee})^{c,\ee}$ gives
\begin{equation*}
\uu_\ee(x) = - \ee \log \int \exp \Big( \frac{v(z) - c(x,z)}{\ee} \Big) d\nu(z).
\end{equation*}
By the chain rule, 
\begin{equation}\label{rappel_v}
\partial_{x_i} \uu_\ee(x) = -\ee \frac{\partial_{x_i} J(x)}{J(x)} 
\quad
\mbox{for}
\quad 
J(x) = \int \exp \Big( \frac{v(z) - c(x,z)}{\ee} \Big) d\nu(z).
\end{equation}
We now turn to the differentiation of $J$.
As shown in \cite{Nutz2022entropic}[Lemma 2.1], 
$$
\inf_{x\in \XX} \{c(x,z) - \uu_\ee(x) \} \leq v(z) \leq \int c(x,z) d\mu_\Sph(x).
$$
By boundedness of $\Sph$, $v$ is bounded and so is the integrand in \eqref{rappel_v}. 
As we deal with probability measures, this justifies using the differentiation under the integral sign, that is 
\begin{equation*}
\partial_{x_i} J(x) = \int \partial_{x_i} \exp \Big( \frac{v(z) - c(x,z)}{\ee} \Big) d\nu(z),
\end{equation*}
which induces
\begin{equation}\label{partialderJ}
\partial_{x_i} J(x) = \int -\frac{\partial_{x_i} c(x,z)}{\ee} \exp \Big( \frac{v(z) - c(x,z)}{\ee} \Big) d\nu(z).
\end{equation}
Fix $z\in \Sph$, so 
\begin{equation}\label{partialderC}
\partial_{x_i} c(x,z) =  d(x,z)\partial_{x_i} d(x,z)
\quad
\mbox{and}
\quad
\partial_{x_i} d(x,z) = \frac{-1}{\sqrt{1 - \langle x,z\rangle^2}}z_i.
\end{equation}
Combining \eqref{partialderJ} with \eqref{partialderC}, 
\begin{equation}\label{partialderJ2}
\partial_{x_i} J(x) = \int \frac{z_i}{\ee} \frac{d(x,z)}{\sqrt{1-\langle x,z\rangle^2}} \exp \Big( \frac{v(z) - c(x,z)}{\ee} \Big) d\nu(z).
\end{equation}
Plugging \eqref{partialderJ2} in \eqref{rappel_v} gives \eqref{grad_u_ee_c_ee}, where $g_\ee$ defined in \eqref{def_ge} shows up because 
\begin{align}
\exp \Big( \frac{\uu_\ee(x) - c(x,z) + \uu_\ee^{c,\ee}(z)}{\ee} \Big) =  \frac{\exp \Big( \frac{\uu_\ee^{c,\ee}(z) - c(x,z)}{\ee} \Big)}{
\int \exp \Big( \frac{\uu_\ee^{c,\ee}(y) - c(x,y)}{\ee} \Big) d\nu(y)}.
\end{align}
The same arguments on $\uu_\ee^{c,\ee}(z)$ yield \eqref{grad_u_ee}.
\hfill
\demend

\subsection{Proof of Proposition \ref{entropicQeeFee}}

Using the Euclidean derivatives obtained in Proposition \ref{entropicmaps},
$$
\nabla \uu_\ee(x) = \rho_x  \int z g_\ee(x,z) d\nu(z)
\quad
\mbox{and}
\quad
\nabla \uu_\ee^{c,\ee}(z) = \rho_z  \int x g_\ee(x,z) d\mu_\Sph(x).
$$
But this is equivalent to 
$$
\nabla \uu_\ee(x) =  \int -\frac{d(x,z)}{\sqrt{1-\langle x,z\rangle^2}} \rho_x(z) 
\exp \Big( \frac{\uu_\ee(x) - c(x,z) + \uu_\ee^{c,\ee}(z)}{\ee} \Big) d\nu(z),
$$
and
$$
\nabla \uu_\ee^{c,\ee}(z) = \int -\frac{d(x,z)}{\sqrt{1-\langle x,z\rangle^2}} \rho_z(x)
\exp \Big( \frac{\uu_\ee(x) - c(x,z) + \uu_\ee^{c,\ee}(z)}{\ee} \Big) d\mu_\Sph(x).
$$
There, one recovers $\Log_x=\Exp_x^{-1}$ as in 
\begin{equation}\label{Logx}
    \Log_x(z) = \frac{d(x,z)}{\sqrt{1-\langle x,z \rangle^2}}\rho_x z = d(x,z) \frac{\rho_x (z-x)}{\Vert \rho_x (z-x) \Vert},
\end{equation}
which gives the result.
\hfill
\demend

\section{Second-order differentiation of entropic potentials}

This section is dedicated to the proof of Proposition \ref{uu_ee_continu_deriv}, that states that the series of spherical harmonics of $\uu_\ee$ belongs to $\ell_1$.

\subsection{Euclidean second-order derivatives}
Note that one can find in \cite{ferreira2014concepts} expressions for first-order and second-order derivatives of the cost $c$. The next proposition gives second-order partial derivatives of the potential $\uu_\ee$.

\begin{prop}\label{Hessian_uu_ee}
The potential $\uu_\ee$ is twice-differentiable everywhere, and 
$$
\frac{\partial^2 \uu_\ee}{\partial_{x_i} \partial_{x_j}}(x) 
= \int \tilde{c}_{ij}(x,z) \exp \Big( \frac{\uu_\ee(x) - c(x,z) + \uu_\ee^{c,\ee}(z) }{\ee} \Big) d\nu(z)
+ \frac{\partial_{x_i}\uu_\ee(x) \partial_{x_j}\uu_\ee(x)}{\ee},
$$
where 
\begin{align*}
    \tilde{c}_{ij}(x,z)
    &= 
    \frac{\partial^2 c(x,z)}{\partial_{x_i} \partial_{x_j}} 
    - \frac{\partial_{x_i}c(x,z)\partial_{x_j}c(x,z)}{\ee}.
\end{align*}
Besides, the same holds for $\uu_\ee^{c,\ee}$ and
$$
\frac{\partial^2 \uu_\ee^{c,\ee} }{\partial_{z_i} \partial_{z_j}}(z) 
= \int \tilde{c}_{ij}(z,x) \exp \Big( \frac{\uu_\ee(x) - c(x,z) + \uu_\ee^{c,\ee}(z) }{\ee} \Big) d\nu(z)
+ \frac{\partial_{z_i}\uu_\ee^{c,\ee}(z) \partial_{z_j}\uu_\ee^{c,\ee}(z)}{\ee}.
$$
\end{prop}

\begin{proof}
The same calculus can be found in \cite{genevay:tel-02458044}[Lemma 3], up to the fact that one can recognize partial derivatives of $\uu_\ee$, that is \eqref{grad_u_ee}, in the result, at the very end of our proof. 
First of all, $g_\ee$ is bounded by using \cite{Nutz2022entropic}[Lemma 2.1] and the compacity of $\Sph$. Thus, one can differentiate in \eqref{grad_u_ee} under the integral sign, and 
\begin{equation}\label{partial2u_ee_0}
\frac{\partial^2 \uu_\ee}{\partial_{x_i} \partial_{x_j}}(x) = \int \partial_{x_j} z_i g_\ee(x,z) d\nu(z).
\end{equation}
In view of using classical rules of differentiation, note that $z_i g_\ee(x,z) = (\partial_{x_i}c(x,z)) G(x,z)$, for 
$$
G(x,z) = \exp \Big( \frac{\uu_\ee(x) - c(x,z) + \uu_\ee^{c,\ee}(z) }{\ee} \Big).
$$
Besides, 
$$
\partial_{x_j} G(x,z) = \frac{1}{\ee} G(x,z) \partial_{x_j}(\uu_\ee(x) - c(x,z)).
$$
As a byproduct, 
\begin{align}
    \partial_{x_j} z_i g_\ee(x,z) &= 
\frac{\partial^2 c(x,z)}{\partial_{x_i} \partial_{x_j}} G(x,z)
+ \partial_{x_i}c(x,z) \frac{1}{\ee} G(x,z) 
\Big( \partial_{x_j}\uu_\ee(x)- \partial_{x_j}c(x,z) \Big),\\
&= G(x,z) \Big( 
\frac{\partial^2 c(x,z)}{\partial_{x_i} \partial_{x_j}}  
+ \partial_{x_i}c(x,z) 
\frac{\partial_{x_j}\uu_\ee(x)- \partial_{x_j}c(x,z)}{\ee}\Big).
\label{derivIntegrandpartial2}
\end{align}
Plugging \eqref{derivIntegrandpartial2} in \eqref{partial2u_ee_0}
and using that $z_i g_\ee(x,z) = (\partial_{x_i}c(x,z)) G(x,z)$ when rearranging, 
\begin{align*}
    \frac{\partial^2 \uu_\ee}{\partial_{x_i} \partial_{x_j}}(x) 
    &= \int G(x,z) \Big(
    \frac{\partial^2 c(x,z)}{\partial_{x_i} \partial_{x_j}} 
    - \frac{\partial_{x_i}c(x,z)\partial_{x_j}c(x,z)}{\ee} \Big)
+ \frac{\partial_{x_j}\uu_\ee(x)}{\ee} z_i g_\ee(x,z) d\nu(z),\\
&= \int G(x,z) \Big(
    \frac{\partial^2 c(x,z)}{\partial_{x_i} \partial_{x_j}} 
    - \frac{\partial_{x_i}c(x,z)\partial_{x_j}c(x,z)}{\ee} \Big) d\nu(z)
+ \frac{1}{\ee}\partial_{x_i}\uu_\ee(x) \partial_{x_j}\uu_\ee(x).
\end{align*}
where we also used the explicit derivatives of $\uu_\ee$ from
\eqref{grad_u_ee}. 
The Hessian of $\uu_\ee^{c,\ee}$ follows by symmetry. 
\end{proof}

\subsection{Proof of Proposition \ref{uu_ee_continu_deriv}}
From Proposition \ref{Hessian_uu_ee}, $\uu_\ee$ is twice-differentiable everywhere, that gives us the continuity of $\QQ_\ee$, and of $\partial_{x_i}\uu_\ee$.
In the expression of the second-order partial derivatives given in \eqref{Hessian_uu_ee}, the term $\frac{1}{\ee}\partial_{x_i}\uu_\ee(x)\partial_{x_j}\uu_\ee(x)$ is thus continuous. 
The remaining term takes the form of a parameter-dependant integral, whose integrand is continuous and bounded. 
Thus, the result follows by a direct application of the theorem for continuity under the integral sign, and by using the property that a sequence of spherical harmonics belongs to $\ell_1$ for functions that are twice continuously differentiable, see \cite{kalf1995expansion}[Theorem 2].

\bibliographystyle{siam}
\bibliography{Spherical.bib}

\end{document}